\documentclass[runningheads]{llncs}
\usepackage{amsmath}
\usepackage{amssymb}
\usepackage{latexsym}
\usepackage[normalem]{ulem}
\usepackage{color}
\usepackage{graphicx}
\usepackage[loose]{subfigure}
\usepackage{cite}

\usepackage[ruled,vlined,linesnumbered]{algorithm2e}
\DontPrintSemicolon
\definecolor{commentgreen}{rgb}{0, 0.5, 0}
\newcommand{\algcom}[2][0.7]{%
  \hfill\makebox[#1\linewidth][l]{\color{commentgreen}\textit{// #2}}%
}
\newcommand{\algcomline}[1]{%
  \hspace{0pt}%
  \begingroup%
  \color{commentgreen}%
  \it%
  // #1%
  \endgroup%
}

\usepackage{soul}  % needs package color

\usepackage[normalem]{ulem}
\def\<#1|#2>{{\color{red}\sout{#1}} {\color{blue}\textbf{#2}}}

\newcommand{\NP}{\ensuremath{\mathcal{NP}}}
\def\O(#1){\ensuremath{\mathcal{O}(#1)}}

\title{Recognizing Optimal 1-Planar Graphs in Linear Time%
  \thanks{Supported by the Deutsche Forschungsgemeinschaft (DFG), grant
    Br835/18-1.}}
\titlerunning{Recognizing Optimal 1-Planar Graphs in Linear Time}

\author{Franz J. Brandenburg}

\institute{University of Passau \\
           94030 Passau \\
           Tel.: +49/851/509-3031 \\
           Fax: +49/851/509-3032\\
           \email{brandenb@informatik.uni-passau.de}
}

\begin{document}
\maketitle
\begin{abstract}
  A graph with $n$ vertices is 1-planar if it can be drawn in the plane such that each edge is
  crossed at most once, and is optimal if it has the maximum of $4n-8$
  edges.

  We show that optimal 1-planar graphs can be recognized in linear time. Our
  algorithm implements a graph reduction system with two rules,
  which can be used to reduce
  every optimal 1-planar graph  to an irreducible extended wheel graph. The
graph reduction system is non-deterministic, constraint, and
non-confluent.

%\keywords{graph drawing \and 1-planar graphs \and recognition
%algorithms
%    \and graph reduction systems }
\end{abstract}

\section{Introduction}

There has been recent interest in beyond planar graphs  that extend
planar graphs by restrictions on crossings. A particular example is
1-planar graphs, which were introduced by Ringel \cite{ringel-65}
and appear when a planar graph and its dual are drawn
simultaneously. A graph is 1-planar if it can be drawn in the plane
with at most one crossing per edge. In his introductory paper on
1-planar graphs, Ringel studied the coloring problem and observed
that a pair of crossing edges can be completed to $K_4$ by adding
planar edges. 1-planar graphs generalize 4-map graphs, which are the
graphs of adjacencies of nations of a map \cite{cgp-mg-02,
cgp-rh4mg-06}. Two nations are adjacent if they share a common
border or if there is a quadripoint where four countries meet, which
results in a $K_4$ in the 4-map graph.

The first study of structural properties of 1-planar graphs is by
Bodendiek, Schumacher, and Wagner \cite{bsw-bs-83,bsw-1og-84}. They
showed that 1-planar graphs with $n$ vertices have at most $4n-8$
edges and that there are such graphs for $n = 8$ and for all $n \geq
10$, and not for $n \leq 7$ and $n = 9$. They called 1-planar graphs
with $4n-8$ edges \emph{optimal} and observed that optimal 1-planar
graphs can be obtained from planar 3-connected quadrangulations by
adding a pair of crossing edges in each quadrangular face. In fact,
this is a characterization and a basis of our recognition algorithm.

As usual,   graphs are simple without self-loops and multiple edges,
and paths and cycles are simple, too. The degree of a vertex is the
number of incident edges or neighbors, and the \emph{local degree}
is the number of incident edges or neighbors when restricted to a
particular induced subgraph. 1-planar graphs are special concerning
their density, which is taken as an upper bound on the number of
edges in relation to the number of vertices. There are three
versions. A 1-planar graph $G$  is \emph{maximally dense} or
\emph{maximum} \cite{s-rm1pg-10} if there is no 1-planar graph  of
the same size with more edges. It is \emph{maximal
  1-planar} if the addition of any edge destroys 1-planarity and
  \emph{planar maximal} or triangulated \cite{cgp-rh4mg-06} if no
  further edge can be added without introducing a crossing.
  Clearly, maximally dense 1-planar graphs are maximal, which in turn
are optimal, but not conversely. Suzuki \cite{s-rm1pg-10} gave all
maximally dense graphs that are not optimal, namely, the complete
graphs for $n \leq 6$, $K_7-2e$ and six graphs with $9$ vertices and
$27$ edges.
 Brandenburg et al.
\cite{begghr-odm1p-13} showed that there are sparse maximal 1-planar
graphs with only $\frac{45}{17}n - \frac{84}{17}$ edges, which is
less than the $3n-6$ bound for maximal planar graphs. Such sparse
maximal 1-planar graphs have many vertices of degree two, whereas
optimal 1-planar graphs have degree at least six \cite{bsw-1og-84}.
Clearly, every maximal planar graph is planar maximal 1-planar,
however, a planar edge can be added to $K_5-e$ if $K_5-e$ is drawn
with a pair of crossing edges. Note that the terms planar maximal,
maximal, maximally dense, and optimal coincide for planar graphs.

An \emph{embedding} (drawing) $\mathcal{E}(G)$ of a graph is a
mapping of $G$ into the plane such that the vertices are mapped to
distinct points and the edges to simple Jordan curves between the
endpoints. It is \emph{planar} if (the Jordan curves of the) edges
do not cross and \emph{1-planar} if each edge is crossed at most
once. An embedding is a witness for planarity and 1-planarity,
respectively. For an algorithmic treatment, a planar embedding is
given by a rotation system, which describes the cyclic ordering of
the edges incident to each vertex, or by the sets of vertices,
edges, and faces. A 1-planar embedding $\mathcal{E}(G)$ is given by
an embedding of the planarization of $G$, which is obtained by
taking the crossing points of edges as virtual vertices
\cite{ehklss-tm1pg-13b}.

A planar embedding of a planar graph can be computed in linear time
as part (or extension) of a planarity test algorithm, see
\cite{p-pte-13}. Accordingly, 1-planarity of an embedding can be
tested in linear time via the planarization. However, computing a
1-planar embedding of a 1-planar graph is \NP-hard. The relationship
between planar graphs and their embeddings is well-understood.
Every 3-connected planar graph has a unique embedding on the sphere
and in the plane if the outer face is fixed \cite{w-pg-33}. The set
of all embeddings of a planar graph can be computed in linear time
and is stored in a  $SPQR$-tree \cite{dt-olpt-96, gm-ltist-01}.
Accordingly, one often uses a planar graph and one of its embeddings
interchangeably.

A 1-planar embedding partitions the edges into \emph{planar} and
\emph{crossing} edges. We color the planar edges black and the crossing ones
red. Other color schemes were used in \cite{ehklss-tm1pg-13b, el-racg1p-13,
  help-ft1pg-12, d-ds1pgd-13}. The \emph{black} or \emph{planar skeleton}
$P(\mathcal{E}(G))$ consists of the black edges and inherits its embedding
from the given 1-planar embedding. Vertex $u$ is called a \emph{black (red)
  neighbor} of vertex $v$ if the edge $(u,v)$ is black (red) in a 1-planar
embedding. A \emph{kite} is a 1-planar embedding of $K_4$ with a pair of
crossing edges and no other vertices in the inner (or outer) face defined by
the black edges. A $K_4$ has one planar and four non-planar embeddings which
differ by the edge coloring and the rotation system \cite{Kyncl-09}, see
Fig.\ \ref{K4-drawings}.

 1-planar embeddings are quite flexible, as the five
embeddings of $K_4$   \cite{Kyncl-09} and the $\NP$-hardness proof
of \cite{abgr-1pgrs-15} show. There is an extension of Whitney's
theorem by Schumacher \cite{s-s1pg-86} who proved that every
5-connected optimal 1-planar graph has a unique 1-planar embedding
with the exception of the extended wheel graphs, which have two
embeddings for graphs of size at least ten and six for the minimum
optimal 1-planar graph   with eight vertices. The extended wheel
graphs $XW_{2k}$   will be described in Section \ref{sect:prelim}.
Suzuki \cite{s-rm1pg-10} improved this result and dropped the
$5$-connectivity precondition.

\begin{figure}
   \centering
   \subfigure[planar]{
%      \rotatebox{180}{%
     \includegraphics[scale=0.3]{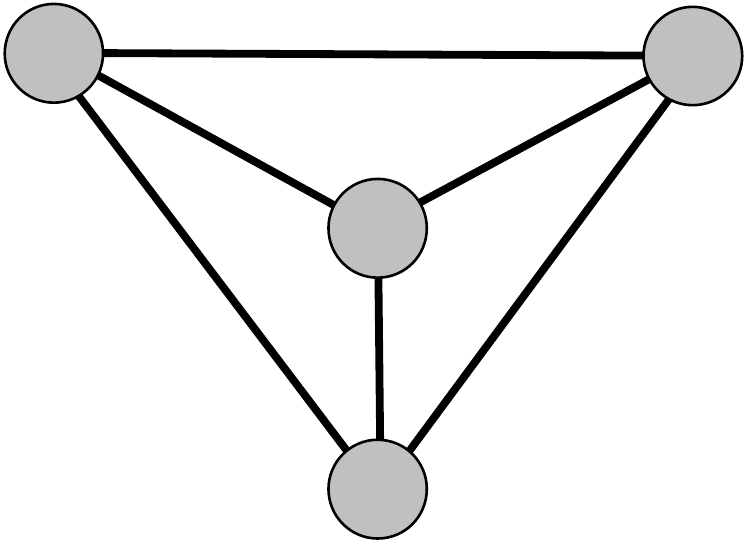}
   }
%   }
   \quad
   \subfigure[crossed as a kite]{
   \hspace{0.75cm}
     \includegraphics[scale=0.3]{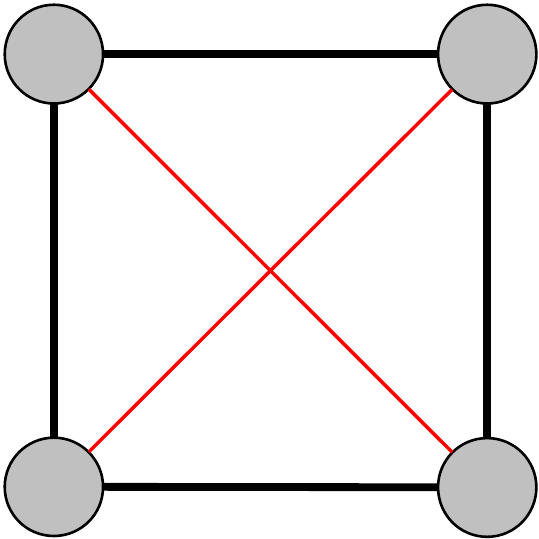}
     \hspace{0.75cm}
   }
   \caption{1-planar embeddings of $K_4$ }
   \label{K4-drawings}
\end{figure}

\iffalse

 Two 1-planar embeddings are said to be
\emph{equivalent} if one can be transformed into the other by a
topological homeomorphism. A 1-planar graph has a \emph{unique
embedding} (into the sphere) if there is only one equivalence class
of 1-planar embeddings. 1-planar embeddings are quite flexible, as
the five embeddings of $K_4$ show. Unique embeddings are well
understood for planar graphs by Whitney's theorem \cite{w-pg-33},
which states that every 3-connected planar graph has a unique
embedding (up to the choice of the outer face in the plane).
Schumacher \cite{s-s1pg-86} proved a related result for 5-connected
optimal 1-planar graphs. He showed that every such graph has a
unique 1-planar embedding with the exception of the extended wheel
graphs, which have two embeddings for graphs of size at least ten
and six for the minimum optimal 1-planar graph $XW_6$ with eight
vertices. The extended wheel graphs $XW_{2k}$ will be described
later. Suzuki \cite{s-rm1pg-10} improved this result and dropped the
5-connectivity precondition, which is a restriction, since optimal
1-planar graphs are 4-connected and not necessarily 5-connected.
Note that a 1-planar graph is 6-connected if it is 5-connected
\cite{s-o1p-84}.

\fi

 A serious drawback of (most classes of) beyond planar graphs is
the general \NP-hardness of their recognition. For 1-planarity this
was proved by Grigoriev and Bodlaender \cite{GB-AGEFCE-07} and by
Korzhik and Mohar \cite{km-mo1ih-13}, and improved to hold for
graphs of bounded bandwidth, pathwidth, or treewidth
\cite{bce-pc1p-13}, for near planar graphs \cite{cm-aoepl-13}, and
for 3-connected 1-planar graphs with a given rotation system
\cite{abgr-1pgrs-15}. Moreover, the recognition of right angle
crossing graphs (RAC) \cite{abs-RACNP-12} and of fan-planar graphs
\cite{BGDMPT-fan-14, bcghk-rfpg-14} is \NP-hard.
On the other hand, Eades et al. \cite{ehklss-tm1pg-13b} introduced a
linear time testing algorithm for (planar) maximal 1-planar graphs
that are given with a rotation system. As aforesaid, 1-planarity of
an embedding can be tested in linear time. In addition, there are
linear time recognition algorithms if all vertices are in the outer
face. The resulting graphs are called  outer 1-planar  and were
first studied by Eggleton \cite{e-rdg-86}. It is not obvious that
outer 1-planar graphs are planar \cite{abbghnr-o1p-15}.
Independently, Auer at al.\ \cite{abbghnr-o1p-15}
  and Hong et al.\   \cite{heklss-ltao1p-15}
developed linear time recognition algorithms for outer 1-planar
graphs. Also, maximal outer-fan-planar graphs can be recognized in
linear time \cite{bcghk-rfpg-14}.
Chen et al. \cite{cgp-rh4mg-06} developed a cubic-time recognition
algorithm for hole-free 4-map graphs and observed  that the
3-connected planar maximal 1-planar graphs are exactly the
3-connected hole-free 4-map graphs \cite{cgp-mg-02}. The optimal
1-planar graphs are exactly the hole-free 4-map graphs with $4n-8$
edges  and thus recognizable in cubic time. Recently, Brandenburg
\cite{b-4m1pg-15} showed that maximal and planar maximal 1-planar
graphs can be recognized in $O(n^5)$ time.

Schumacher \cite{s-s1pg-86} defined a single-rule graph
transformation system on 1-planar embeddings and proved that every
$5$-connected optimal 1-planar graph is reducible to an extended
wheel graph  which is irreducible. His result was generalized by
Suzuki \cite{s-rm1pg-10} who added a second rule and thereby removed
the $5$-connectivity restriction. The reduction rules are defined on
an embedding and   extend the reduction rules for planar
$3$-connected   quadrangulations of Brinkmann et al.\
\cite{bggmtw-gsqs-05}.

 In this paper  we translate the reduction rules of Schumacher and Suzuki
 from 1-planar embeddings to 1-planar graphs and  show how to
 implement them efficiently. In consequence, the proof of existence
 for a reduction of an optimal 1-planar graph to an irreducible extended wheel graph
by Schumacher \cite{s-s1pg-86} and Suzuki \cite{s-rm1pg-10} is
transformed into an efficient algorithm. These proofs say  that a
($5$-connected) graph $G$ is optimal 1-planar if and only if there
exists a natural number $k$ and a computation by a sequence of
applications of reductions such that an extended wheel graph
$XW_{2k}$ is obtained from $G$, in symbols, $G \rightarrow^*
XW_{2k}$. Suzuki reverses direction and expands $XW_{2k}$ into $G$.
Again, one must guess the start $k$ or $XW_{2k}$ and the expansion
process.

We show that the usability of a reduction rule can be checked in
$O(1)$ time on graphs.  According to Brinkmann et. al.\
\cite{bggmtw-gsqs-05}, a feasible use of a reduction must preserve
  the given class, i.e., the optimal 1-planar
graphs. Thereby, we obtain a simple quadratic-time recognition
algorithm of optimal 1-planar graphs which is improved to a linear
time algorithm by a bookkeeping technique. It can be extended to
maximally dense 1-planar graphs and specialized to $5$-connected
optimal 1-planar graphs. Our algorithm improves upon the cubic
running time algorithm of Chen et al. \cite{cgp-rh4mg-06}, which
solves a more general problem and searches $4$-cycles and other
types of separators. Combinatorial properties of the  reductions
  are explored in \cite{b-rso1p-16}.

The paper is organized as follows: In the next Section we recall
some basic properties of optimal 1-planar graphs. In Section
\ref{sect:rules} we introduce  the reductions rules and show how to
apply them to graphs.   The  linear recognition algorithm for
optimal 1-planar graphs is established in Section
\ref{sect:lineartime}, and we conclude with some open problems on
1-planar graphs.

\section{Preliminaries} \label{sect:prelim}

Optimal 1-planar graphs have   special properties. Schumacher
\cite{s-s1pg-86} observed that there is a one-to-one correspondence
between optimal 1-planar graphs and their planar skeletons  which
are 3-connected quadrangulations. An optimal 1-planar graph is
obtained from a 3-connected quadrangulation by adding a pair of
crossing edges in each quadrilateral face to form a kite. Thus the
red edges are added to the black ones. A formal proof was given by
Suzuki \cite{s-rm1pg-10}. All vertices of an optimal 1-planar graph
have an even degree of at least six and there are at least eight
vertices of degree six, since in total there are $4n-8$ edges if the
given graph has $n$ vertices. The planar and the crossing edges
alternate in the rotation system of a 1-planar embedding of an
optimal 1-planar graph. Consider, for example, graph $B_{17}$ in
Fig.\ \ref{fig:B17} which has $17$ vertices, $60$ edges and an even
degree of at least six at each vertex. Is $B_{17}$ optimal 1-planar?

\begin{figure}
   \begin{center}
     \includegraphics[scale=0.5]{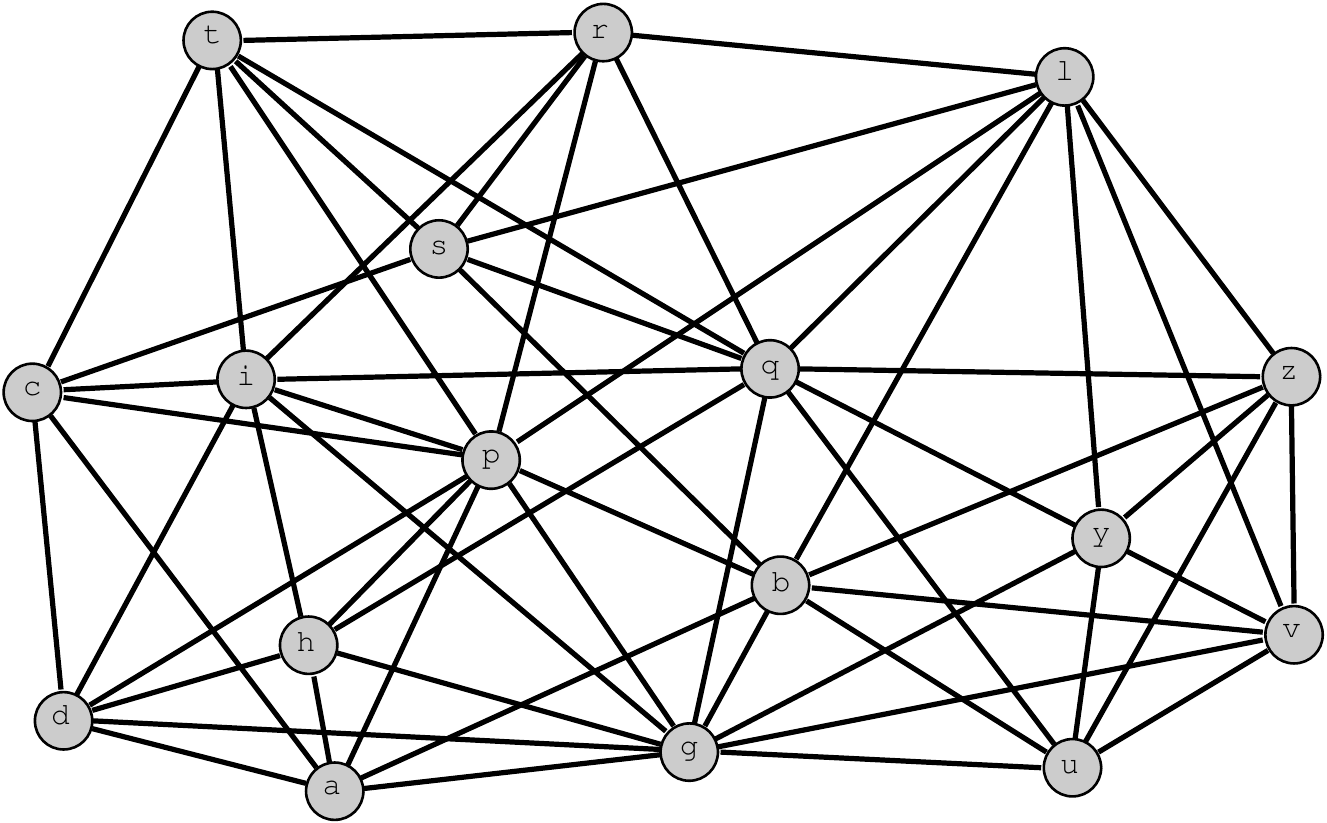}
     \caption{A candidate graph $B_{17}$ with $17$ vertices and $60$ edges
     \label{fig:B17}}
   \end{center}
\end{figure}

The exact number of  optimal 1-planar graphs is known for graphs of
size up to $36$. Bodendiek et al. \cite{bsw-1og-84} showed that
$K_6$ is 1-planar but is not optimal and that there are no optimal
1-planar graphs with seven and nine vertices.  There is a unique
optimal 1-planar graph for $n=8,10,11$, and there are three optimal
1-planar graphs for $n = 12, 13$. For $n=14$, they found $11$
optimal 1-planar graphs, but one is missing.
  Brinkmann et al.
\cite{bggmtw-gsqs-05} developed recurrence relations for the
enumeration of quadrangulations and computed the number of
3-connected quadrangulations up to size $36$. For example, there are
12 for $n=14$ and 3000183106119 quadrangulations and  optimal
1-planar graphs of size 36.

The \emph{pseudo-double wheels} \cite{bggmtw-gsqs-05} and the
\emph{extended wheel graphs} $XW_{2k}$ play a particular role for
quadrangulations and   optimal 1-planar graphs, respectively, since
thes are the irreducible or minimum graphs under two graph reduction
rules.
For $k\geq 3$, a pseudo-double wheel $W_{2k}$ is a quadrangulation
with two distinguished vertices $p$ and $q$, called \emph{poles},  a
cycle of even length with vertices $v_1, \ldots, v_{2k}$   and edges
$(v_i, v_{i+1})$ in circular order
 and further edges $(p, v_{2i})$
and $(q, v_{2i-1})$ for $i=1,\ldots, k$. Thus $p$ is connected with
the vertices at even and $q$ with the vertices
 at odd positions on the cycle. $W_{2k}$ has $n=2k+2$ vertices,
 $2n-4$ edges and $n-2$ faces.
The extended wheel graph $XW_{2k}$ additionally contains all
possible pairs of 1-planar crossing edges $(p, v_{2i-1}), (v_{2i},
v_{2i+2})$ and $(q, v_{2i}), (v_{2i-1}, v_{2i+1})$ in circular
order. This is the augmentation of $W_{2k}$ by kites,   see Fig.\
\ref{XWgraphs}. The two poles of $XW_{2k}$ have degree  $2k$ and
each of the $2k$ vertices on the cycle has degree six.  If $k \geq
4$, then the edges $(v_i, v_{i+1})$ on the cycle are black and the
edges $(v_{2i}, v_{2i+2})$ and $(v_{2i-1}, v_{2i+1})$ are red.
In addition, a graph is an extended wheel graph if it is optimal
1-planar and has a vertex of degree $n-2$ \cite{bsw-1og-84}. The
second degree $n-2$ vertex is implied. Moreover, an optimal 1-planar
graph is an extended wheel graph if the vertices of degree six form
a cycle  \cite{bsw-1og-84}.

The notation $XW_{2k}$ for graphs of size $2k+2$ is taken from
Suzuki \cite{s-rm1pg-10} and is related to Schumacher's $2 *
\hat{C}_{2k}$ notation.

\begin{figure}
   \centering
   \subfigure[The minimum extended wheel graph $XW_6$ drawn as a crossed
       cube]{
     \parbox[b]{4.5cm}{%
       \centering
       \includegraphics[scale=0.45]{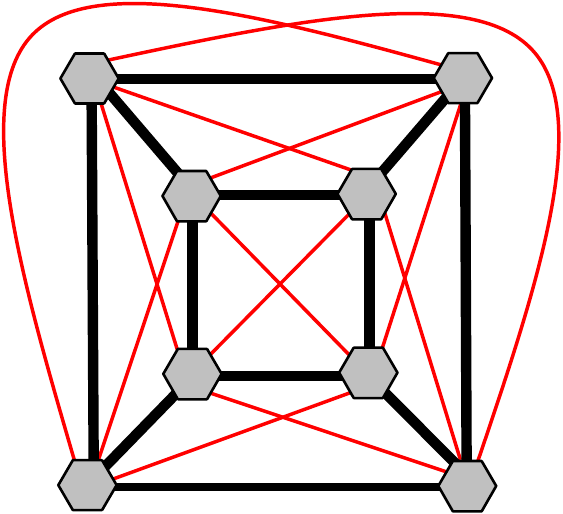}
     }
   }
   \hfil
   \subfigure[The extended wheel graph $XW_{10}$ with poles $p$ and $q$ and
       with hexagons for vertices of degree $6$]{
     \parbox[b]{6.1cm}{%
       \centering
       \includegraphics[scale=0.33]{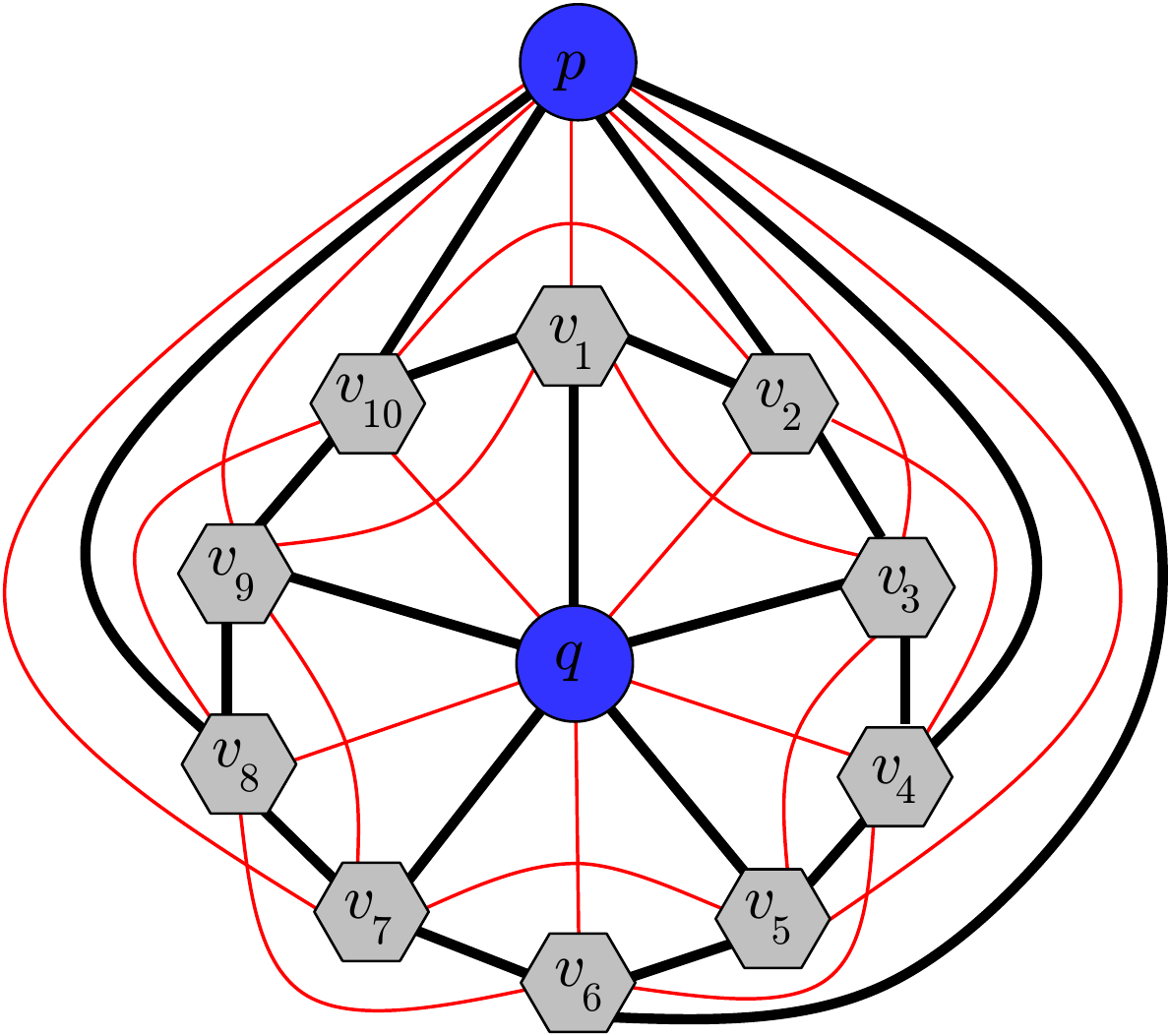}
     }
   }
   \caption{Extended wheel graphs $XW_6$ and $XW_{10}$. Any two non-adjacent
     vertices $p$ and $q$ of $XW_6$ can be taken as poles.
     In larger extended wheel graphs, if poles $p$ and $q$ change places this
     swaps the coloring of the incident edges. Edges between consecutive
     vertices on the cycle are always planar and are colored black.}
   \label{XWgraphs}
\end{figure}

We summarize some basic properties of optimal 1-planar graphs from
\cite{bsw-bs-83,bsw-1og-84, s-rm1pg-10}.

\begin{proposition} \label{prop:characterization} Every optimal 1-planar
  graph $G = (V,E)$ consists of a planar quadrangulation $G_P = (V,E_P)$ and
  a pair of crossing edges   $e,f$ in each face of $G_P$
  forming a kite such  that $E = E_P \cup E_C$, where $E_C$ is the set of crossing edges.
$G_P$ is  3-connected and bipartite. $G$ has a
  unique embedding, except if $G$ is an extended wheel graph $XW_{2k}$,
  which has two inequivalent embeddings for $k \geq 4$ in which the planar
  and crossing edges incident to a pole are interchanged and their
  colors swap. The minimum extended wheel graph $XW_6$ has six inequivalent
  1-planar embeddings.
\end{proposition}

\noindent From the fact that  $G_P$ is bipartite, we can conclude:

\begin{lemma} \label{circles}
Every cycle of odd length in an optimal
  1-planar graph contains at least one red edge. If $C$ is a cycle of
  length four and three of its edges are black, then all edges of $C$ are
  black.
\end{lemma}

Schumacher \cite{s-s1pg-86} defined a relation  on 1-planar
embeddings and used it to characterize $5$-connected optimal
1-planar graphs.

\begin{definition}
Two 1-planar embeddings $\mathcal{E}(G)$ and $\mathcal{E}(G')$ are
\emph{related}, $\mathcal{E}(G) \hookrightarrow \mathcal{E}(G')$, if
there is a planar quadrangle $(v_1, v_2, v_3, v_4)$ in the planar
skeleton
$P(\mathcal{E}(G))$ such that\\
 \centerline{ (*)\, \,\,  all paths from $v_1$ to
$v_3$ of length four in $P(\mathcal{E}(G))$  pass through $v_2$ or
$v_4$.}
\noindent Then $\mathcal{E}(G')$ is obtained from
$\mathcal{E}(G)$ by merging $v_1$ and $v_3$ and removing parallel
edges. For graphs $G$ and $G'$,  let $G \hookrightarrow G'$ if there
exist embeddings such that $\mathcal{E}(G) \hookrightarrow
\mathcal{E}(G')$ and denote the transitive closure by
``$\hookrightarrow^*$''.
\end{definition}

 The paths of (*) from
$v_1$ to $v_3$ are simple and use only planar (black) edges. The
embedding $\mathcal{E}(G)$ must satisfy special properties such that
the planar quadrangle $(v_1, v_2, v_3, v_4)$ coincides with $(x,
x_3, x_4, x_5)$ in Fig.\ \ref{SR}. Note that each quadrangle in an
extended wheel graph has a path of length four between opposite
vertices $(v_{i-1}, v_{i+1})$ of a planar quadrangle through one of
the poles, such that (*) is violated. In consequence, the
``$\hookrightarrow$''-relation is not applicable.

\begin{proposition} \cite{s-s1pg-86}  \label{prop:properties}
  Every 5-connected optimal 1-planar graph $G$ can be reduced to an extended
  wheel graph $XW_{2k}$ for some $k \geq 3$, i.e., $G \hookrightarrow^*
  XW_{2k}$. The extended wheel graphs are irreducible (or minimum) elements
  under the ``$\hookrightarrow$''-relation.
\end{proposition}

By the restriction to 5-connected graphs, Schumacher excluded
graphs with separating 4-cycles.  Separating 4-cycles play a similar
role in optimal 1-planar graphs as separating triangles do in
triangulated planar graphs. In fact,  every non-irreducible
$5$-connected optimal 1-planar graph can be
reduced to $XW_8$ \cite{b-rso1p-16}.\\

Brinkmann et al. \cite{bggmtw-gsqs-05} introduced two graph
transformations, called $P_1$-~  and $P_3$-expansions, for the
generation and characterization of  (planar) 3-connected
quadrangulations. We consider their inverse as reductions.

\begin{definition}
The $P_1$-\emph{reduction} on a quadrangulation consists of a
contraction of a face $f = (u,x,v,z)$ at $x,z$, where $x$ has degree
$3$ and $u,v,z$ have degree at least $3$. It is shown in Fig.\
\ref{fig:splitting} and in an augmented version in Fig.\ \ref{SR}
with the restriction to planar (black) edges. The
$P_3$-\emph{reduction} removes the vertices of the inner cycle of a
planar cube, where the inner cycle is empty and the vertices of the
outer cycle have degree at least $4$, see Fig.\ \ref{CR} restriced
to planar (black) edges.

The reductions must be applied such that they preserve the class of
 3-connected   quadrangulations.
\end{definition}

By the one-to-one correspondence between 3-connected
quadrangulations and optimal 1-planar graphs, the $P_1$-~  and
$P_3$-reductions are extended straightforwardly to embedded 1-planar
graphs, called vertex  and face contraction  by Suzuki
\cite{s-rm1pg-10}. Their inverse is called $Q_v$-splitting and
$Q_4$-cycle addition, respectively, and are used from right to left.
The illustration in Fig.\ \ref{fig:splitting} is taken from
\cite{s-rm1pg-10}. A $Q_4$-cycle addition removes the pair of
crossing edges of a kite and inserts five  new kites as illustrated
in Fig.\ \ref{CR}.   Suzuki \cite{s-rm1pg-10} observed that
Schumacher's ``$\hookrightarrow$''-relation  coincides with his face
contraction
 and defines the $P_1$-reduction on the planar
skeleton of an embedded 1-planar graph.

\begin{figure}
   \begin{center}
   \rotatebox{270}{%
     \includegraphics[scale=0.35]{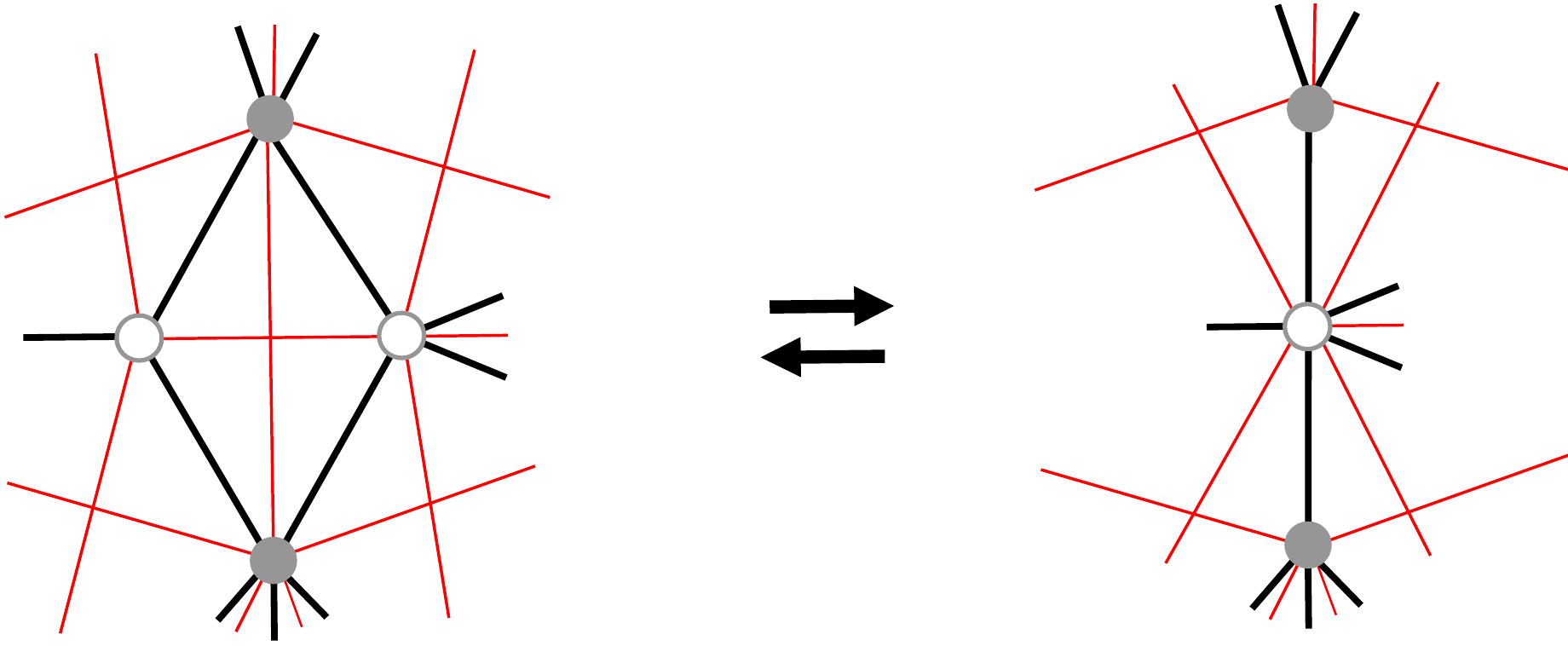}
     }
     \caption{Face contraction or  vertex splitting on 1-planar embeddings
     with planar (black and thick) and crossing (red and thin) edges
     \label{fig:splitting}}
   \end{center}
\end{figure}

 The
distinction between graphs and embeddings is not  important for the
$P_1$-~ and $P_3$-reductions of Brinkmann et. al., since there is a
one-to-one correspondence on $3$-connected planar graphs. They point
out that the reductions must be used with care such that the given
class of graphs is preserved. It is not specified, however, how this
is achieved. On the other hand, the ``$\hookrightarrow$''-relation
of Schumacher and the $Q_v$-splitting and $Q_4$-cycle addition and
the inverse $Q_f$-contraction and $Q_4$-removal of Suzuki need a
1-planar embedding and the distinction between planar (black) and
crossing (red) edges. It is not immediately clear how to apply these
rules to graphs that are given without an embedding or an edge
coloring. Nevertheless they characterize the respective graphs, as
stated in Propositions \ref{prop3} and \ref{prop4}.

\section{Reduction Rules and Their Application} \label{sect:rules}

For the translation of the reduction rules from embeddings to graphs
and an efficient check of their usability, we use the uniqueness of
1-planar embeddings of reducible optimal 1-planar graphs and the
local environment of a reduction. In consequence, a reduction is
applied to a subgraph  which has (almost) a fixed embedding. A
primary goal is to compute the embedding  and  to check  the
feasibility of the application of a reduction. The correctness
follows from the works of Brinkmann et al. \cite{bggmtw-gsqs-05},
Schumacher \cite{s-s1pg-86}, and Suzuki \cite{s-rm1pg-10}.

Transformations on graphs and graph replacement systems have been
studied in the theory of graph grammars \cite{r-hgg-97}. In general,
a graph transformation is a pair of left-hand and right-hand side
graphs $\alpha = (L, R)$. An application of $\alpha$ to a graph $G$
replaces an occurrence of $L$ in $G$ by an occurrence of $R$ while
the remainder $G-L$ is preserved. It results in a graph $G' =
G-L+R$. A graph $L$ occurs in $G$ and $L$ is said to \emph{match} a
subgraph $H$ of $G$ if there is a graph homomorphism between $L$ and
$H$, which is one-to-one and onto on the vertices and one-to-one but
not necessarily onto for the edges, and similarly for $R$ and $G'$.
Unmatched edges of $H$ remain in $G-L$ and are kept for $G-L+R$.
This is elaborated in the algebraic approach to graph
transformations \cite{cmrehl-97}. In this particular case, the
general approach does not really help, since the complexity of the
element problem of graph grammars is PSPACE hard \cite{b-ccgg-83}.

 We reverse the expansions of Brinkmann et al. and Suzuki and
call them \emph{SR-reduction} (Schumacher reduction) and
\emph{CR-reduction} (crossed cube reduction), and the graphs of the
left-hand sides $CS$ (crossed star) and $CC$ (crossed cube),
respectively. The $SR$-reduction augments the vertex splitting of
Suzuki and includes the subgraph induced by the center $x$. The
reductions are shown in Figs. \ref{SR} and \ref{CR} including a
1-planar embedding and an edge coloring. The tiny strokes at the
outer vertices indicate further edges, which are necessary. These
vertices may have even more edges to outer vertices.

Following Brinkmann et al. \cite{bggmtw-gsqs-05}, the given class,
here the optimal 1-planar graphs, must be preserved and therefore an
application of a reduction is \emph{constrained}. An infeasible
application may destroy the 3-connectivity of the underlying planar
skeleton or introduce multiple edges, which ultimately leads to a
violation of 3-connectivity.

\begin{figure}
   \begin{center}
     \includegraphics[scale=0.40]{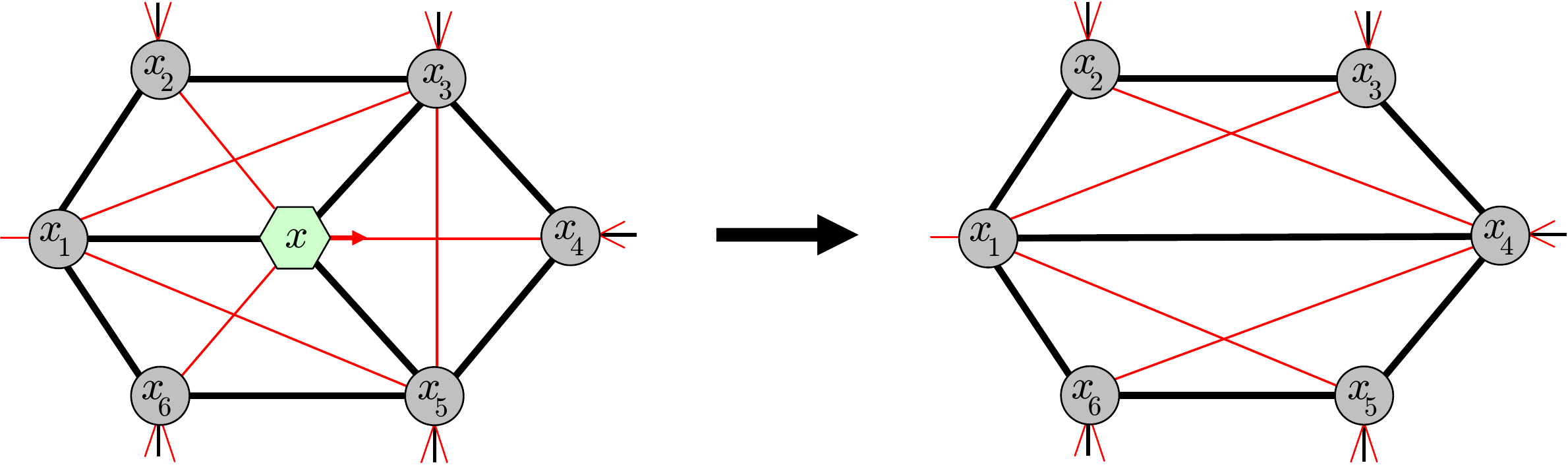}
     \caption{The reduction $SR(x  \mapsto x_4)$ for optimal 1-planar
       graphs. A candidate is drawn as a hexagon and other vertices as circles.
       Good candidates are light green   and bad ones orange. Planar edges
       are drawn black and thick and crossing edges red and thin. The
       tiny strokes at the outside indicate further necessary edges. The left
       graph is $CS$ together with its embedding.\label{SR}}
   \end{center}
\end{figure}

\begin{figure}
   \begin{center}
     \includegraphics[scale=0.32]{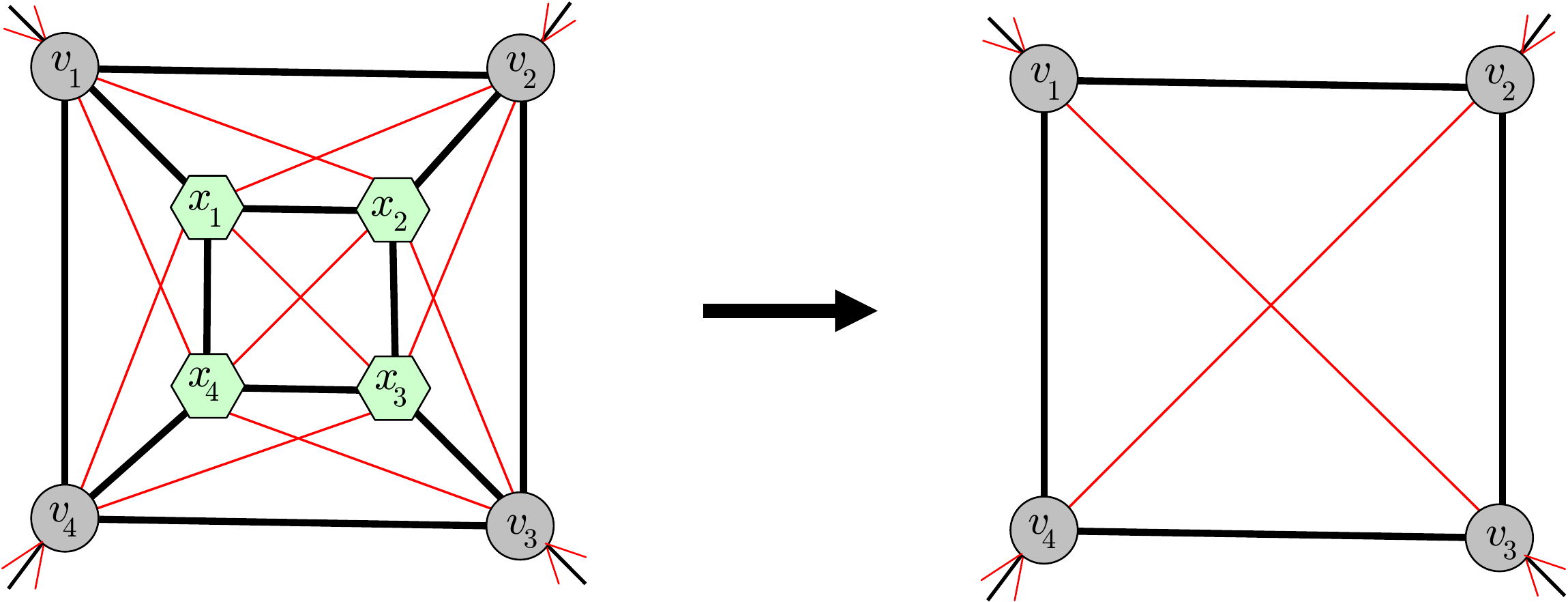}
     \caption{The reduction $CR(x_1, x_2, x_3, x_4)$ for optimal 1-planar
       graphs with candidates $x_1, x_2, x_3, x_4$. The reduction is good and the candidates are colored light
       green
       if there is no  edge $(v_1, v_3)$ or $(v_2, v_4)$ in the outer face. The
       graph on the left is $CC$ together with its embedding.\label{CR}}
   \end{center}
\end{figure}

 The main task of our algorithm is an efficient and feasible use
of the  reduction rules  such that optimal 1-planar graphs are
preserved. An obstacle is the gap between the graphs $CS$ and $CC$
of the left-hand side of the reductions, which come with an
embedding, and the matched subgraph $H$, which comes as a part of
$G$. A \emph{matched subgraph} $H(x)$ of a $SR$-reduction is a
subgraph induced by a vertex $x$ of degree six and its neighbors.
There are three red and three black neighbors which alternate in the
circular order around $x$. For a $CR$-reduction there is a subgraph
of eight vertices. A matched subgraph may have further edges, since
the matching is not onto for the edges. This introduces so-called
blocking edges and is discussed later on.
 We grant 3-connectivity of the underlying planar
skeleton by the absence of a \emph{blocking vertex}, which has
degree six. For example, vertices $x_3$ or $x_5$ are blocking
vertices of $SR(x \mapsto x_4)$ in Fig.\ \ref{SR} if they have
degree six. In case of a $CR$-reduction, a  vertex is blocking if it
is matched by a vertex from the outer cycle of $CC$ and has degree
six. Multiple edges are avoided
 by the absence of \emph{blocking edges}, which
can be planar or crossing, i.e., black or red. A blocking edge is
always related to a reduction and it may be blocking for many
reductions. Blocking black edges occur  in separating 4-cycles, and
 red and black  blocking edges are treated differently.
 Blocking vertices and planar blocking
edges can also appear in the planar case, whereas  blocking red
edges are exclusive to 1-planar graphs. They also cover the case of
blocking vertices, since a blocking vertex implies a  blocking red
edge. The converse does not hold.

 A matching of $CS$ or $CC$ with a
subgraph $H$ shall classify the edges of $H$ as planar and crossing
and color them black and red, respectively. It shall determine  the
circular order of the vertices in the outer face of $CS$ and $CC$,
 and thus an embedding of $H$. However, this is not always the
case. Graph $CS$ has several 1-planar embeddings, since some $K_4$'s
may be drawn planar or as a kite. In fact, $CS$ is a planar graph,
however, as a subgraph of a 1-planar graph it must be embedded with
crossings as shown in Fig.\ \ref{SR}, since reducible optimal
1-planar graphs have a unique embedding. Furthermore,  if the
matched graph  of $CS$ also has   edges $(x_2, x_6)$ and $(x_3,
x_6)$, then it has two 1-planar embeddings in which $x_1$ and $x_6$
may change places, which implies a color change of the incident
edges, just as in the case of  extended wheel graphs. If edges
$(x_2, x_4), (x_2, x_6$ and $(x_4, x_6)$ exists in addition to the
edges of $CS$, then the situation is even worse and any circular
order of the neighbors of $x$ is possible. Fortunately, these
possibilities are represented by the degree vectors which are
defined below.

The usability of a reduction is linked to one or four vertices of
degree six and some conditions. A  $SR$-reduction is applied to a
vertex $x$ of degree six, which is the image of the central vertex
and the corner of three kites of $CS$. For the right-hand side, $x$
is merged with a target, which is a red vertex $v$ of the outer
cycle, denoted $SR(x \mapsto v)$, and $SR(x \mapsto x_4)$ is shown
in Fig. \ref{SR}. A given optimal 1-planar graph
 may have several places for the application of a reduction, even at a single candidate,
 and
 the next reduction is chosen nondeterministically. There are candidates where a reduction is feasible
and others where a reduction is   infeasible. An application of a
$CR$-reduction is linked to (one of) four vertices $x_1, x_2, x_3,
x_4$ of degree six, which are all infeasible for a  $SR$-reduction,
and is denoted $CR(x_1, x_2, x_3, x_4)$. The vertices are on the
inner cycle of $CC$ and are removed and replaced by a pair of
crossing edges, such that the vertices from the outer cycle form a
kite. In  a drawing, the inner cycle may be at the outside.

For convenience, we say that $SR$ is applied to vertex $x$ of the
given graph if $SR(x \mapsto v)$ is feasible and call $v$   the
\emph{target} of $x$, and similarly, that $CR$ is applied to $(x_1,
x_2, x_3, x_4)$ or just to $x_i$ for some $i=1,2,3,4$. In addition,
we shall identify the vertices and edges of the left-hand sides $CS$
or $CC$ with those of the matched subgraph $H$, although the
embedding and edge coloring of $H$ is not yet fixed and some
vertices might change places. In general, the matching and embedding
will be clear. Sometimes, it would be good to increase the degree of
a vertex $u$, e.g., to avoid that $u$ is a blocking vertex for
another reduction. The simplest way is to apply the inverse of $CR$,
i.e., the $Q_4$-cycle addition of \cite{s-rm1pg-10}, and insert a
new $4$-cycle together with five pairs of crossing edges in a
quadrangular face at $u$ that is left if a pair of crossing edges is
removed.

\begin{definition}
  A vertex $x$ (of an optimal 1-planar graph $G$) of degree six is called a
  \emph{candidate}.

  A candidate $x$ is \emph{``good''} if there is a \emph{feasible application} of a
  reduction at $x$. Then the reduction is used as given in Figs.\ \ref{SR}
  and \ref{CR} and the class of optimal 1-planar graphs is
  preserved.
  A good candidate  is drawn as a light green hexagon. In case of $SR$, $x$ is
  the center of a subgraph $H(x)$ of $G$ that matches $CS$ and there is some
  red neighbor $v$, called a \emph{target},  such that $SR$ can be applied by merging $x$
  with $v$, denoted $SR(x \mapsto v)$. Then $SR(x \mapsto v)$ is \emph{``good''}
  and can feasibly be \emph{applied} to $x$.
  There are three targets  in $H(x)$ for a $SR$-reduction.
  In case of a $CR$-reduction, vertex
  $x$ belongs to the inner cycle of a subgraph $H$ that matches $CC$, and
  $CS$ is applied to any vertex of the inner cycle.

Otherwise, $x$ is a \emph{``bad''} candidate  and is drawn as an
orange hexagon. Then the reductions $SR(x \mapsto v)$ are bad for
all three red neighbors   of $x$. The usage is illegal.
  A bad reduction $SR(x \mapsto v)$ is
  \emph{blocked} by a vertex $u$ if $u$ is a black neighbor of $x$ and $v$
  of degree six and if $u$ is any vertex on the outer cycle of degree six in case of
  $CR$, respectively.
  An edge $e=(u, v)$ of $H(x)$ is a \emph{blocking red edge} of
  $SR(x \mapsto v)$ if $u$ is a red neighbors of $x$. Edge $e$ is a
  \emph{blocking black edge} if $u$ is a black neighbor and
  $e$ is not matched by an edge of $CS$.
If the outer cycle of $CC$ matches $(v_1,v_2, v_3, v_4)$, then edges
$(v_1, v_3)$ and $(v_2, v_4)$ are \emph{blocking red edges} of a
$CR$-reduction.
\end{definition}

A subgraph $H(x)$ with neighbors $(x_1, \ldots, x_6)$ of $x$ in
circular order  may have up to three blocking red edges, namely
$(x_2, x_4), (x_4, x_6)$ and $(x_6, x_2)$ if $x_2, x_4, x_6$ are the
red neighbors of $x$, see Fig.\ \ref{fig:blockingedges}. There may
be none. Blocking red edges are associated in pairs with
$SR$-reductions, and each blocking red edge $(u,v)$ is associated
with two $SR$-reductions, $SR(x \mapsto u)$ and $SR(x \mapsto v)$.
The edges must be red by Lemma \ref{circles}. Accordingly, a
blocking black edge $(u,v)$ of $SR(x \mapsto v)$ connects $v$ with
the vertex at the opposite side of $CS$, since it is not matched,
and, again, it must be black by Lemma \ref{circles}, see Fig.\
\ref{fig:blockingedges}.  There are up to three blocking black edges
in $H(x)$, and each $SR$-reduction has at most one blocking black
edge, since blocking black edges  do not cross. There are two
blocking red  edges in case a $CR$-reduction, however, at most one
of them can occur  in an optimal 1-planar graph that is not the
minimum extended wheel  graph $XW_6$.   By Lemma \ref{circles},
blocking black edges are excluded in this case.

\begin{figure}
   \centering
   \subfigure[Graph $CS$ with three blocking red edges $e,f$ and $g$. There may be a
       subgraph in the area between $e$ and (x2, x3, x4) and similarly for $f$ and $g$. $SR(x \mapsto x_4)$ is blocked
       by $x_3$ if there is no such subgraph and $e$ is crossed by a red edge incident to $x_3$.]{
     \parbox[b]{6.5cm}{%
       \centering
       \includegraphics[scale=0.33]{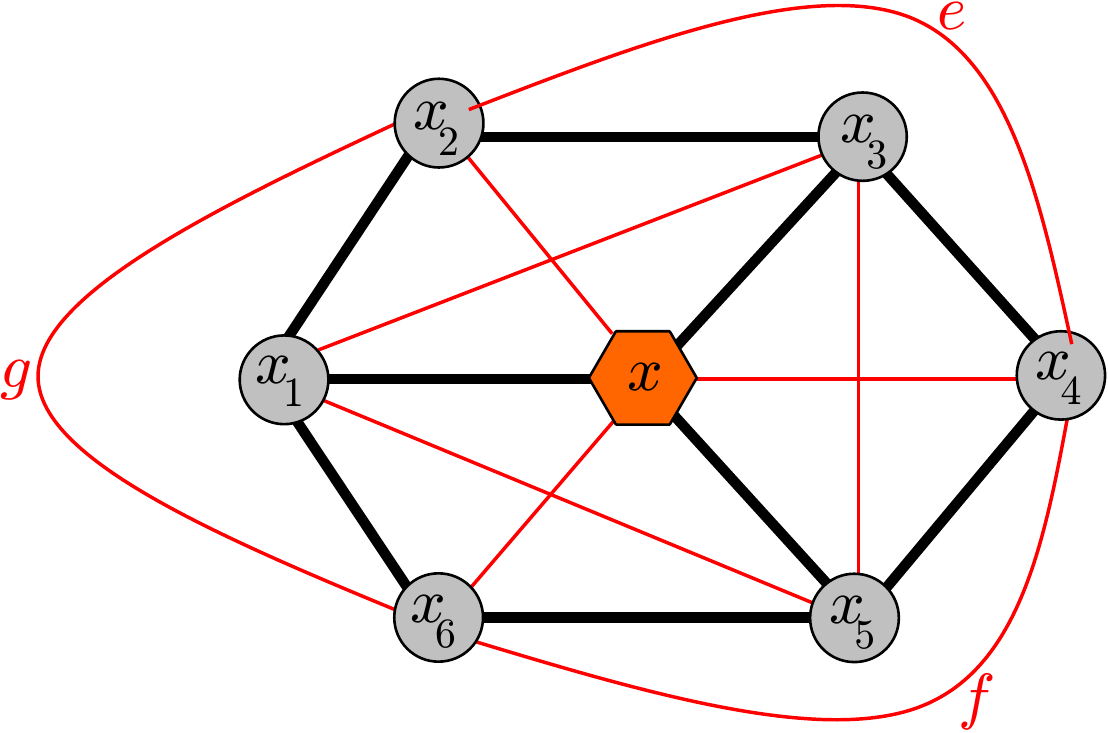}
     }
   }
   \hfil
   \subfigure[Graph $CS$ with a blocking black edge and two blocking red
       edges]{
     \parbox[b]{4.3cm}{%
       \centering
       \includegraphics[scale=0.33]{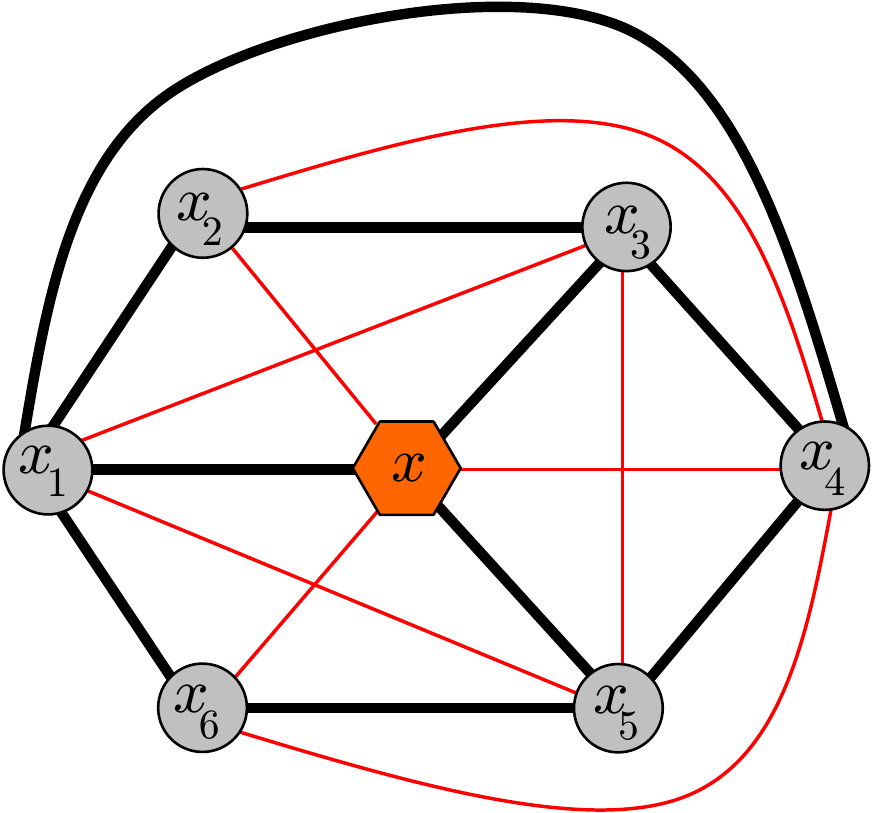}
     }
   }
   \caption{Blocking edges added to $CS$}
   \label{fig:blockingedges}
\end{figure}

An application of a reduction  with a blocking vertex would decrease
the degree of the blocking vertex to four, which would violate the
3-connectivity of the planar skeleton. The resulting graph would no
longer be optimal 1-planar. The application of a reduction with a
blocking edge would introduce a multiple edge, whose endpoints are a
separation pair of the planar skeleton. This again leads to a
violation of the 3-connectivity of the planar skeleton. Note that
the case of a blocking vertex is covered by a blocking red edge
between the black neighbors of the blocking vertex on the outer
cycle. The converse is not true, since the blocking edge may enclose
a (larger) subgraph. For example, add a forth vertex and then apply
the inverse of $CR$.

\begin{example} \label{ex1}
Consider  graph $G_{17}$  which  is optimal 1-planar by the 1-planar
embedding displayed in Fig.\ \ref{fig:base}. Vertices $a,c,d,h,
r,s,t, u,v,y,z$ are candidates, where $a, c,h,s,t$ are good for a
$SR$-reduction, whereas $d$ and $r$ are bad candidates, and
therefore are colored orange. Vertices $u,v,y,z$ are good for a
$CR$-reduction. A good $SR$-reduction $SR(x \mapsto v)$ is indicated
by an arrowhead on the red edge from $x$ to $v$. Vertex $x$ moves
along that edge and is merged with $v$. For example, $SR(a \mapsto
b)$ is good, whereas $SR(a \mapsto c)$  and $SR(a \mapsto h)$ are
bad, since $d$ is a blocking vertex and $(c,h)$ is a blocking red
edge.
\end{example}

As another example, consider   extended wheel graphs  as in Fig.\
\ref{XWgraphs}. Every vertex on the cycle of $XW_{2k}$ is a
candidate which, however, is blocked by its neighbors on the cycle.
In addition, all vertices of  $XW_6$  are blocked candidates. Hence,
all $SR$-reductions are bad and the extended wheel graphs are
irreducible.

\begin{definition}
  Every reduction $\alpha$ has a pair $\{e,f\}$ of \emph{associated red blocking edges}. If
 $\alpha = SR(x \mapsto v)$, then  $e = (u, v)$ and $f = (w, v)$,
 where $u, v$ and $w$ are the
  red neighbors of $x$.
  If $\alpha$ is a $CR$-reduction, then $e$ and
  $f$ are the two blocking red edges of $CC$.
\end{definition}

Recall that $SR$ is   Schumacher's ``$\hookrightarrow$''-relation
and the reverse of Suzuki's vertex splitting $Q_v$
\cite{s-rm1pg-10}, and $CR$ is the reverse of the $Q_4$-cycle
addition, and the $P_1$-~  and $P_3$-expansions of Brinkmann et al.\
are the restrictions to planar quadrangulations.
 The existential foundations for the reductions
are:

\begin{proposition} \cite{bggmtw-gsqs-05} \label{prop3}
The class $\mathcal{Q}_3$ of all 3-connected quadrangulations of the
sphere is generated from the pseudo-double wheels by the
$P_1(\mathcal{Q}_3)$-~  and  $P_3(\mathcal{Q}_3)$-expansions.
\end{proposition}

\begin{proposition} \cite{s-rm1pg-10} \label{prop4}
Every  optimal 1-planar graph   can be obtained from an extended
wheel graph by a sequence of $Q_v$-splittings and $Q_4$-cycle
additions.
\end{proposition}

\subsection{Application of the Reduction Rules}

For an application of a reduction  one must find a matching of $CS$
and $CC$ and a subgraph $H$ that preserves the coloring and the
1-planar embedding of $CS$ and $CC$, respectively, and check whether
the reduction is good or bad. In addition, the reason for a bad
reduction must be known for the linear time algorithm in Section
\ref{sect:lineartime}. Fortunately, the degree vector and the local
degrees of the vertices of the matched subgraph provide the
necessary information, as stated in  Table 1.

\begin{definition}
  Let $x$ be a candidate of graph
  $G$, and let $H(x)$ be the subgraph induced by
  $x$ and its six neighbors. Let $\overrightarrow{H(x)} = (d_1, \ldots, d_7)$ be
  the lexicographically ordered 7-tuple with the local degrees of the
  vertices of $H(x)$, called the \emph{degree vector} of $x$,
  and call $\tau(x) = d_1$ the \emph{type} of $x$.
\end{definition}

\begin{lemma}\label{lem:H(x)}
  If $x$ is a candidate of an optimal 1-planar graph, then
  \begin{enumerate}
  \item $3 \leq \tau(x) \leq 5$
  \item $H(x)$ has between $15$ and $18$ edges, and
  \item $\overrightarrow{H(x)} \in \{(3,3,3,5,5,5,6), (3,3,4,5,5,6,6),
    (3,4,4,5,5,5,6), (3,4,5,5,5,6,6)$, \\
    \hspace*{11.5mm} $(4,4,5,5,5,5,6), (4,4,5,5,6,6,6), (5,5,5,5,5,5,6)\}$.
  \end{enumerate}
\end{lemma}

\begin{proof}
  The first tuple for $\overrightarrow{H(x)}$ is the degree vector of $CS$
  and any sparser subgraph cannot match $CS$. As $x$ is the corner of three
  kites, one can add at most three extra edges in the outer face of $CS$,
  namely $(x_2, x_4), (x_2, x_6), (x_6, x_4)$ with $\overrightarrow{H(x)} =
  (5,5,5,5,5,5,6)$ and $(x_2, x_4), (x_2, x_6), (x_2, x_5)$ with
  $\overrightarrow{H(x)} = (4,4,5,5,6,6,6)$, where e.g., $(x_2, x_5)$ must
  be black and the other edges are red. The other degree vectors result from
  one or two edges added to $CS$. \qed
\end{proof}

Obviously, $\overrightarrow{H(a)} = (3,4,4,5,5,5,6)$ for vertex $a$
of $G_{17}$ in Fig.\ \ref{fig:base} and vertex $b$ has local degree
$3$. Moreover, $\overrightarrow{H(x)} = (4,4,5,5,5,5,6)$ if $x$ is
on the inner cycle of $CC$ and the $CR$-reduction is good, such as
$u, v, y,z$ in $G_{17}$, and if $x$ is on the cycle of an extended
wheel graph $XW_{2k}$ for $k \geq 4$. Finally, the maximum degree
vector
 $\overrightarrow{H(x)} = (5,5,5,5,5,5,6)$ appears at every vertex of $XW_6$ and
 at two candidates of $CC$ if there is a
blocking red edge, e.g., $(b,g)$  in Fig. \ref{fig:G17ah}.

\begin{definition}
The subgraph $H(x)$ of a candidate $x$ of an optimal 1-planar graph
is \emph{fixed} if its embedding and coloring is uniquely
determined. It has a \emph{partial coloring} if two neighbors of $x$
may change places and the coloring of the incident edges is open,
and, finally, $H(x)$ is \emph{unclear} if the coloring of the edges
of $H(x)$ is undecided.
\end{definition}

\begin{lemma}\label{localtest}
  Let $x$ be a candidate of an optimal 1-planar graph $G$ and $H(x)$ the
  subgraph induced by $x$ and its neighbors.

  \begin{enumerate}
  \item If $\tau(x) = 3$, then the coloring of $H(x)$ is fixed except for
    $\overrightarrow{H(x)} = (3,4,5,5,5,6,6)$, where there is a
    partial coloring.
  \item If $\tau(x) = 4$, then  there is a
    partial coloring.
  \item If $\tau(x) = 5$, then
    the edge coloring is unclear.
  \end{enumerate}
\end{lemma}

\begin{proof}
First, a black neighbor  of $x$ has local degree at least $5$.

If $\tau(x)=3$ and $x_4$ has local degree $3$, then $x_4$ is a red
neighbor of $x$ and  has two more neighbors, say $x_3$ and $x_5$,
that are black neighbors of $x$ and $x_4$. Then the subgraph induced
by $(x, x_3, x_4, x_5)$ must form a kite, since it is $K_4$ and the
embedding is unique by Proposition \ref{prop:characterization}.
If there is another vertex with local degree $3$, then the above
applies again, such that the circular order of the neighbors of $x$,
the edge coloring  and the embedding  of $H(x)$ are uniquely
determined. If $d_2 = d_3 =4$, then the two vertices with local
degree $4$ are red neighbors of $x$ and they have a red edge in
between, whose removal leaves two  vertices with local degree $3$.
Again, $H(x)$ is uniquely determined.
Finally, consider  $ \overrightarrow{H(x)} = (3,4,5,5,5,6,6)$ with
$x_4$ of local degree $3$ and   $x_2$ of local degree $4$. Vertex
$x_4$ determines $x_3$ and $x_5$ as its black neighbors on the
cycle. There are no edges $(x_2, x_4)$ and $(x_2, x_5)$ such that
$x_2$ is opposite of $x_5$. Vertices $x_2$ and $x_4$ have $x_3$ as
common neighbor  and $x_3$ is a black neighbor of $x, x_2$ and
$x_4$.
 However, the roles of $x_1$ and $x_6$ are
undecided in $H(x)$. They may change places in the circular order
around $x$, but the edge $(x_1, x_6)$ is black, see Fig.
\ref{undecided}.
Thus there is a partial coloring of $H(x)$.

If $\tau(x) = 4$,   there are two vertices of local degree $4$ by
Lemma \ref{lem:H(x)}. Let  $x_2$ and $x_4$ be these vertices, which
are red neighbors of $x$. The third red neighbor of $x$ has local
degree at least $5$. Hence,   edge $(x_2, x_4)$ is missing in
$H(x)$. There is a vertex of local degree $5$ that is not adjacent
to $x_4$ and is opposite of $x_4$ and similarly for $x_2$. Let $x_1$
and $x_5$ be the respective vertices, which are black neighbors of
$x$. Edges $(x_1, x_2)$ and $(x_4, x_5)$ are black, and the subgraph
induced by $\{x, x_1,
  x_3, x_4, x_6\}$ is fixed. However,   $x_3$ and $x_6$ may change
  places and there is a partial edge coloring.
  Finally, the neighbors of $x$ are indistinguishable and the edge
  coloring is unclear if $\tau(x) = 5$.
\qed
\end{proof}

\begin{figure}
  \centering
  \subfigure[The correct 1-planar drawing.]{
    \includegraphics[scale=0.37]{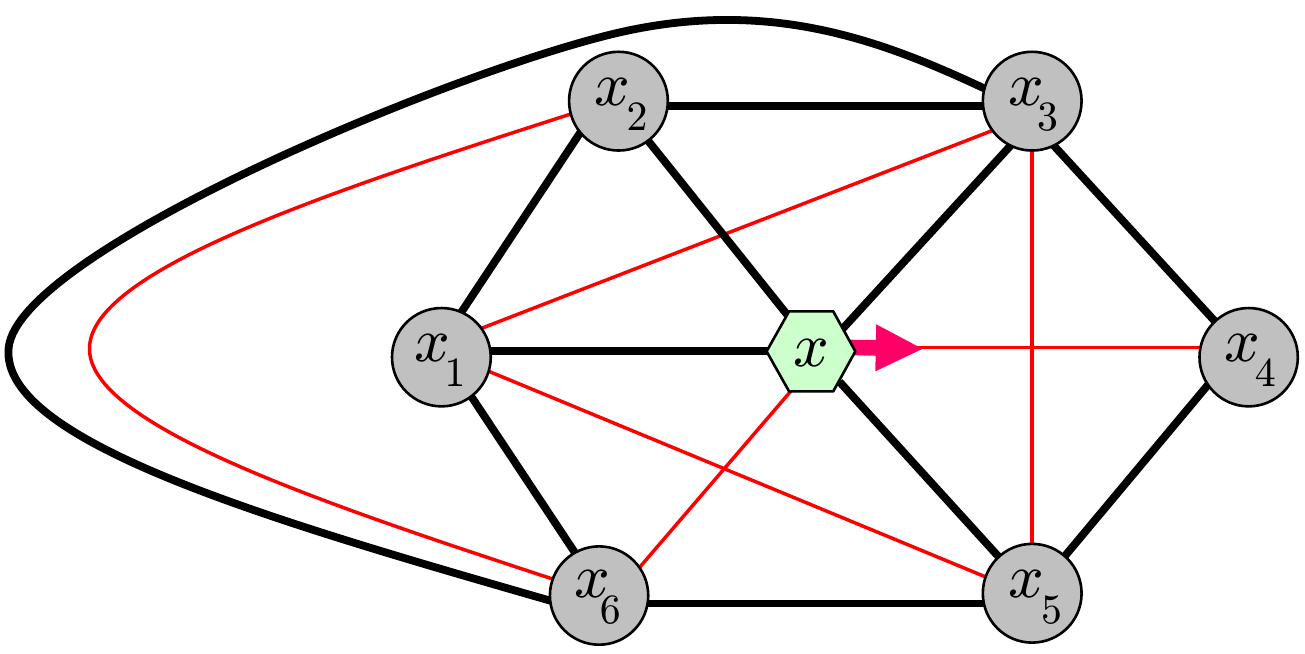}
    \label{undecided-a}
  }
  \hfil
  \subfigure[The incorrect 1-planar drawing with a new edge coloring.
      For a correct drawing swap $x_1$ and $x_6$.]{
     \parbox[b]{6.3cm}{%
       \centering
       \includegraphics[scale=0.37]{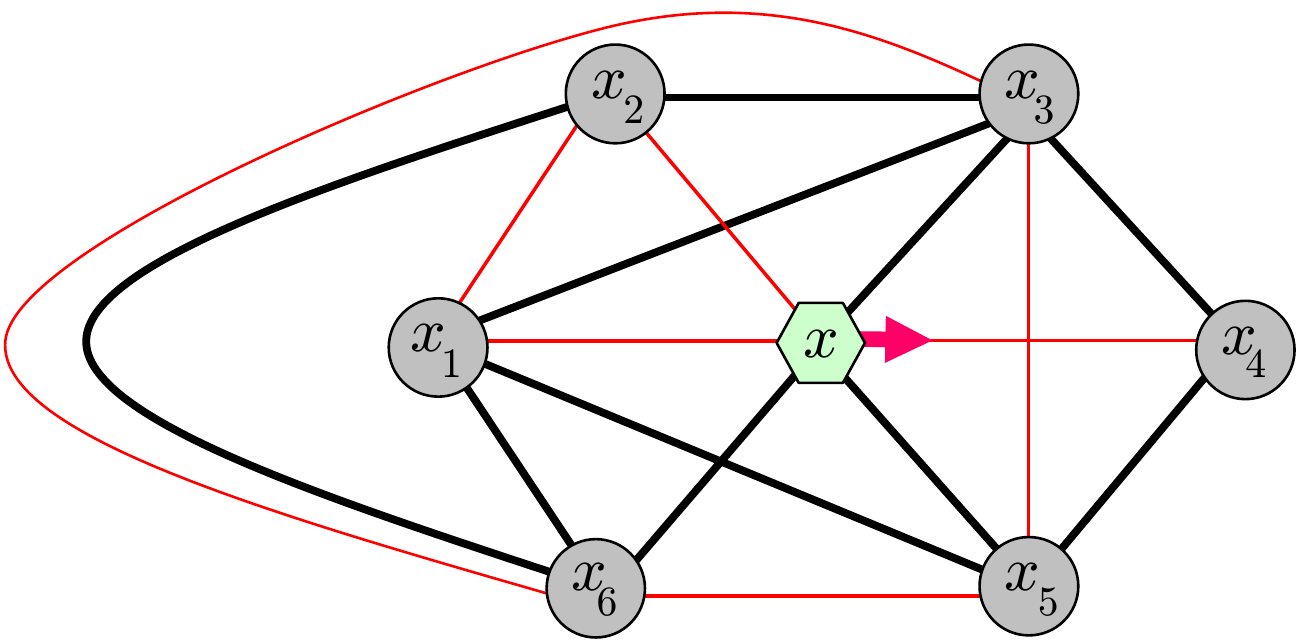}
     }
    \label{undecided-b}
  }
  \caption{An ambiguous case where the embedding of $x_1$ and $x_6$
    is not yet fixed.}
  \label{undecided}
\end{figure}

Fortunately, neighboring candidates help each other in determining
the edge coloring. Consider the candidates $x_1, x_2, x_3, x_4$ of
the inner cycle of $CC$, as given in Fig.\ \ref{CR}, and assume that
the graph is not $XW_6$. Then $\tau(x_i) \geq 4$ and  $\tau(x_i)= 4$
for two of them, say $x_1$ and $x_3$. Then $\overrightarrow{H(x_1)}$
determines that $(x_1, x_2), (x_1, x_4), (x_2, v_2)$ and $(x_4,
v_4)$ are black, whereas $v_1$ and $x_3$  may change places.
Similarly, $\overrightarrow{H(x_3)}$ determines the black edges
$(x_3, x_2), (x_2, v_2),$ $(x_3, x_3),$ $(x_4, v_4)$. If $\tau(x_2)
=5$ then $\tau(x_4)= 5$ and their coloring is unclear. However, the
black neighbors of $x_2$ are $x_1, v_2, x_3$ and $v_1$ is a red
neighbor, which implies that $v_1$ is a black neighbor of $x_1$ and
the case is decided. Hence, the coloring of a subgraph matching $CC$
is fixed and its embedding is unique.

In fact, we have the following situation:

\begin{lemma} \label{lem:twoH}
Let $H$ be a (sub-) graph of size $7$ with a vertex $x$ of degree
$6$. Then $H$ is unique and is maximal planar if
$\overrightarrow{H(x)} = (3,3,3,5,5,5,6)$. There are (at least) two
graphs $H_1$ and $H_2$ if $\overrightarrow{H(x)} = (3,4,4,5,5,5,6)$,
and $H_1$ and $H_2$ are non-planar and $1$-planar.
\end{lemma}

\begin{proof}
Let $\{x, v_1, \ldots, v_6\}$  be the vertices of $H$, where $x$ has
degree $6$, and   $v_2, v_4, v_6$ have  local degree $3$. If each of
$v_2, v_4, v_6$ has two vertices of $v_1, v_3, v_5$ as neighbors,
then $H = CS$. Clearly, $CS$ satisfies the assumptions and is
planar. For a contradiction, suppose that there is an edge $(v_2,
v_4)$ and let $u, v$ be the two remaining neighbors of $v_2$ and
$v_4$. Then $v_1, v_3$ and $v_5$ cannot have local degree $5$ and
there is no graph as required.

Let $H_1$ be obtained from $CS$ by adding   edge $(x_2, x_4)$, and
let $H_2$ be the graph displayed in Fig.\ \ref{fig:secondH}. There
is an edge from the vertex of degree $3$ to a vertex of degree $4$
in $H_2$, which does not exist in $H_1$. The graphs are non-planar,
since there are $7$ vertices and $16$ edges and they are 1-planar,
as shown by the figures. \qed
\end{proof}

\begin{figure}
   \begin{center}
    \includegraphics[scale=0.55]{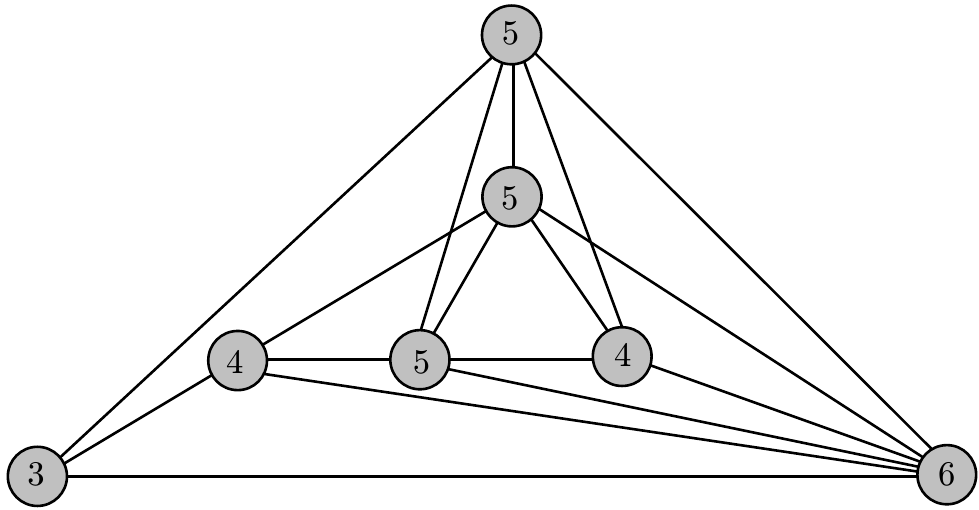}
     \caption{Graph $H_2$ from Lemma \ref{lem:twoH} with vertices labeled by
     degree  \label{fig:secondH}}
   \end{center}
\end{figure}

Similarly, there is a unique subgraph that matches $CC$ if $CR$ is
good. The unique embedding is obtained from pairs of  vertices that
are placed opposite each other on the inner and outer cycles.

\begin{lemma} \label{lem:uniqueCC}
There is a unique graph $H$  that matches $CC$ if $H$ has four
mutually neighbored candidates   $x_1, x_2, x_3, x_4$ with
$\overrightarrow{H(x_i)} = (4,4,5,5,5,5,6)$ for $i=1,2,3,4$.
\end{lemma}

\begin{proof}
Let $v_1, v_2, v_3, v_4$ be the remaining vertices of $H$.  Each
vertex of the inner cycle of $CC$ excludes exactly one vertex of the
outer cycle as a neighbor.
 Assume this does not hold in $H$ and let e.g.,\
 $x_1$ and $x_2$ both exclude $v_1$ as a neighbor.
 Since $x_1$ and $x_2$ are candidates, vertices $v_2, v_3,
v_4$ are their neighbors, and both $x_1$ and $x_2$ contribute a $6$
in the degree vector of the other such that $\overrightarrow{H(x_1)}
= (\ldots, 6,6)$, a  contradiction. \qed
\end{proof}

The usability of a reduction is completely determined by the degree
vector of a candidate and the type distinguishes between a $SR$-~
and a $CR$-reduction.

\begin{lemma}\label{apply-SR}
  A candidate $x$ of an optimal 1-planar graph is good for $SR$ if and only
  if $\tau(x) = 3$ and $SR(x \mapsto v)$ is good if  $v$ has
  local degree $3$.
\end{lemma}

\begin{proof}
Let $H(x)$ be the subgraph matching $CR$ and let $\{x, x_1, \ldots,
x_6\}$ be the vertices of $H(x)$. Then $H(x)$ has a unique embedding
  matching the embedding of $CS$ as shown in Lemma
\ref{localtest} if $\tau(x) = 3$ and $\overrightarrow{H(x)} \neq
(3,4,5,5,5,6,6)$. Then $SR(x \mapsto x_4)$ is good if $x_4$ has
local degree $3$. Otherwise, there is a partial coloring and $x_1$
and $x_6$ may change places if $x_4$ has local degree $3$ and $x_2$
has local degree $4$. This ambiguity does not hinder using $SR(x
\mapsto x_4)$, which removes $x$ and the edge $(x_3, x_5)$ and
inserts the edges $(x_1, x_4), (x_6, x_4)$ and the red edge $(x_2,
x_4)$. Then the color of the edges incident to $x_1$ and $x_6$
remains open.

If $\tau(x) \geq 4$, then every red neighbor of $x$ has a blocking
red edge and there is no good $SR$-reduction. \qed
\end{proof}

\begin{lemma}\label{apply-CR}
  A candidate $x_1$ of an optimal 1-planar graph is good for $CR$ if and only
  if $\tau(x_1) = 4$ and there are three more candidates $x_2, x_3, x_4$
  with $\overrightarrow{H(x_i)} = (4,4,5,5,5,5,6)$, and $CC$ matches the subgraph induced by
  $x_1, x_2, x_3, x_4$ and its four common neighbors.
\end{lemma}
\begin{proof}
If $CR$ is good, then the degree vector of the four vertices that
match the vertices of the inner cycle of $CC$ is $(4,4,5,5,5,5,6)$.
The degree vector $(5,5,5,5,5,6,6)$ implies a blocking red edge and
$\overrightarrow{H(x)} = (4,4,5,5,6,6,6)$ implies a black edge
between two opposite neighbors of a center, which violates $CC$.

Conversely, there is a unique subgraph matching $CC$ by Lemma
\ref{lem:uniqueCC} if the degree vector of the candidates is
$\overrightarrow{H(x)} = (4,4,5,5,6,6,6)$, and there is no blocking
edge. \qed
\end{proof}

\begin{corollary} \label{good-bad}
For each candidate $x$ of a graph $G$ it can be checked in $O(1)$
time whether $x$ is good or bad. It can be determined which
reduction applies if $x$ is good.
 The reduction takes $O(1)$
time including   a (partial) coloring of the edges.
\end{corollary}

\begin{proof}
The type of $x$ decides which reduction may apply and the degree
vector(s) and the local degrees tell whether the reduction is good.
The reductions operate on subgraphs with six resp. eight vertices.
They remove one or four vertices and one more edge and insert three
or two edges. This can be accomplished in $O(1)$ time. \qed
\end{proof}

We summarize the degree vectors and their impact on an edge
coloring,  reductions and their blocking edges, and storing the
reductions in the linear-time algorithm in Section
\ref{sect:lineartime} in Table 1.
 For convenience,  assume that the circular order of the
neighbors of candidate $x$ is $(x_1, \ldots x_6)$ as in Fig.\
\ref{SR}, where $x_2, x_4$ and $x_6$  are red neighbors and $x_4
\leq x_2 \leq x_6$ if the vertices are ordered by local degree.
 Let $e=(x_2, x_4), f = (x_6, x_4), g = (x_2,
x_6)$ and let $d$ and $d'$ be the diagonals $(v_1, v_3)$ and $(v_2,
v_4)$ in case of $CC$ and a
$CR$-reduction.\\

\begin{table}[t]
\begin{center}

\begin{tabular}{c | c | c | c | c }
  $\overrightarrow{H(x)}$ & coloring  & reductions & blocking edges & storing  \\
  \hline \hline

  $(3,3,3,5,5,5,6)$ & fixed & $SR(x \mapsto x_2)$ & none &  $GOOD_e, GOOD_g$ \\
   &   & $SR(x \mapsto x_4)$  & none & $GOOD_e, GOOD_f$ \\
    &   & $SR(x \mapsto x_6)$  & none & $GOOD_f, GOOD_g$ \\

  $(3,3,4,5,5,6,6)$ & fixed & $SR(x \mapsto x_2)$ & none & $GOOD_e, GOOD_g$ \\
    &   & $SR(x \mapsto x_4)$  & none & $GOOD_e, GOOD_f$  \\
    &   & $SR(x \mapsto x_6)$  & black & none \\

   $(3,4,4,5,5,5,6)$ & fixed & $SR(x \mapsto x_2)$ & $g$   & $WAIT_e, BAD_g$ \\
    &   &   $SR(x \mapsto x_4)$  & none & $GOOD_e, GOOD_f$ \\
    &   &   $SR(x \mapsto x_6)$ & $g$ & $WAIT_f, BAD_g$ \\

  $(3,4,5,5,5,6,6)$ & partial & $SR(x \mapsto x_2)$ & $g$ & $WAIT_e, BAD_g$ \\
    &   & $SR(x \mapsto x_4)$ & none & $GOOD_e, GOOD_f$ \\
    &   & $SR(x \mapsto x_6)$ & black & none \\

  $(4,4,5,5,5,5,6)$ & fixed & $CR$ & none & $GOOD_d, GOOD_{d'}$ \\
  $(4,4,5,5,6,6,6)$ & fixed & $CR$ & $d$ & $BAD_d, WAIT_{d'}$ \\
  $(5,5,5,5,5,5,6)$ & unclear & infeasible &   &   \\
  \hline
\end{tabular}
\end{center}
\label{table1}
 \caption{Degree vectors and their impact on reductions}
\end{table}

The existence of a good candidate is granted unless all candidates are
blocked, as in an extended wheel graph, or if the graph is not optimal
1-planar.

\begin{lemma} \label{existence} If $G$ is a reducible  optimal 1-planar graph,
then $G$ has a good candidate.
\end{lemma}

\begin{proof}
  According to Brinkmann et al. \cite{bggmtw-gsqs-05} there is a good
  candidate for their $P_1$-~  and $P_3$-reductions (or expansions) on
  3-connected quadrangulations unless the graph is a double-wheel
  graph, and thus irreducible.
  In  Lemma 4 \cite{bggmtw-gsqs-05} they prove  that a good candidate lies in the innermost
  (or outermost) separating $4$-cycle. By the one-to-one correspondence
  between planar 3-connected quadrangulations and optimal 1-planar graphs,
  this generalizes to optimal 1-planar graphs. \qed
\end{proof}

As a final step, we consider the recognition of extended wheel
graphs.

\begin{lemma} \label{XW-test} There is a linear time algorithm to test
  whether a graph is an extended wheel graph $XW_{2k}$.
\end{lemma}

\begin{proof}
  If the input graph $G$ has eight vertices, we check  $G = XW_{6}$ by
  inspection. Here, each vertex $x$ is a candidate with $\overrightarrow{H(x)} =
  (5,5,5,5,5,5,6)$.

  For $k \geq 4$, an extended wheel graph $XW_{2k}$ has two
  poles $p$ and $q$ of degree $2k$ as distinguished vertices and a cycle
  of $2k$ vertices of degree six. This is checked in a
  preprocessing step on the given graph and takes en passant $O(1)$ time.
  For a final check, we remove the poles and restrict ourselves to the subgraph induced
  by the vertices of degree six. Each such vertex $v$ has four neighbors and
  the cyclic ordering of these vertices is determined as $(v_{-2}, v_{-1},
  v, v_{+1}, v_{+2})$ by the missing edges $(v_{-2}, v_{+1}), (v_{-2},
  v_{+2})$ and $(v_{-1}, v_{+2})$. So we determine the cycle and then check
  for $XW_{2k}$. Altogether, the tests take $O(2k)$ time.
  \qed
\end{proof}

From the above observations, we obtain a simple quadratic-time
algorithm for the recognition of optimal 1-planar graphs. The
algorithm scans the actual graph and  searches a single candidate
for $SR$ or a cluster of four candidates for $CR$ and checks in
$O(1)$ time whether the reduction is good or bad. Each reduction
removes one or four vertices. Hence, there are at most $n-2k-2$
reductions from a graph of size $n$ to an extended wheel graph
$XW_{2k}$.

\begin{theorem} \label{thm:quadratic}
There is a quadratic-time recognition algorithm for optimal 1-planar
graphs.
\end{theorem}

\begin{example} \label{ex:example1}
For an explanation of the reductions consider the input graph
$G_{17}$ as shown in Fig.\ \ref{fig:base} with a 1-planar embedding.
Vertices   $a, c,h,s,t$ are good for an $SR$-reduction, and
$u,v,y,z$ are good for a $CR$-reduction. If the $CR$-reduction is
applied first, we obtain the graph in Fig.\ \ref{fig:G17CR} and
$SR(h \mapsto q)$ then yields $XW_{10}$.

Alternatively, using $SR(a \mapsto b), SR(h \mapsto q), SR(g \mapsto
d), SR(c \mapsto b), SR(d \mapsto i)$ and finally $CR(u,v,y,z)$ ends
up at $XW_6$. This computation is illustrated in Figs.\
\ref{fig:G17a} to \ref{fig:G17ahgcdxy}.  Note that the cluster
$u,v,y,z$ flips from good to bad if  there is an outer neighbor of
degree six, which is blocking and induces a blocking red edge.

Also, $XW_8$ can be obtained by $SR(a \mapsto b), SR(h \mapsto q),
SR(g \mapsto d)$, and finally $CR(u,v,y,z)$.
\end{example}

\begin{figure}
  \centering
  \subfigure[A 1-planar embedding of $G_{17}$. Candidates are drawn as
      hexagons which are light green for good candidates and  orange for bad
      candidates. Non-candidates of degree at least $8$ are drawn as
      circles.]{
     \parbox[b]{\textwidth}{%
       \centering
      \includegraphics[scale=0.35]{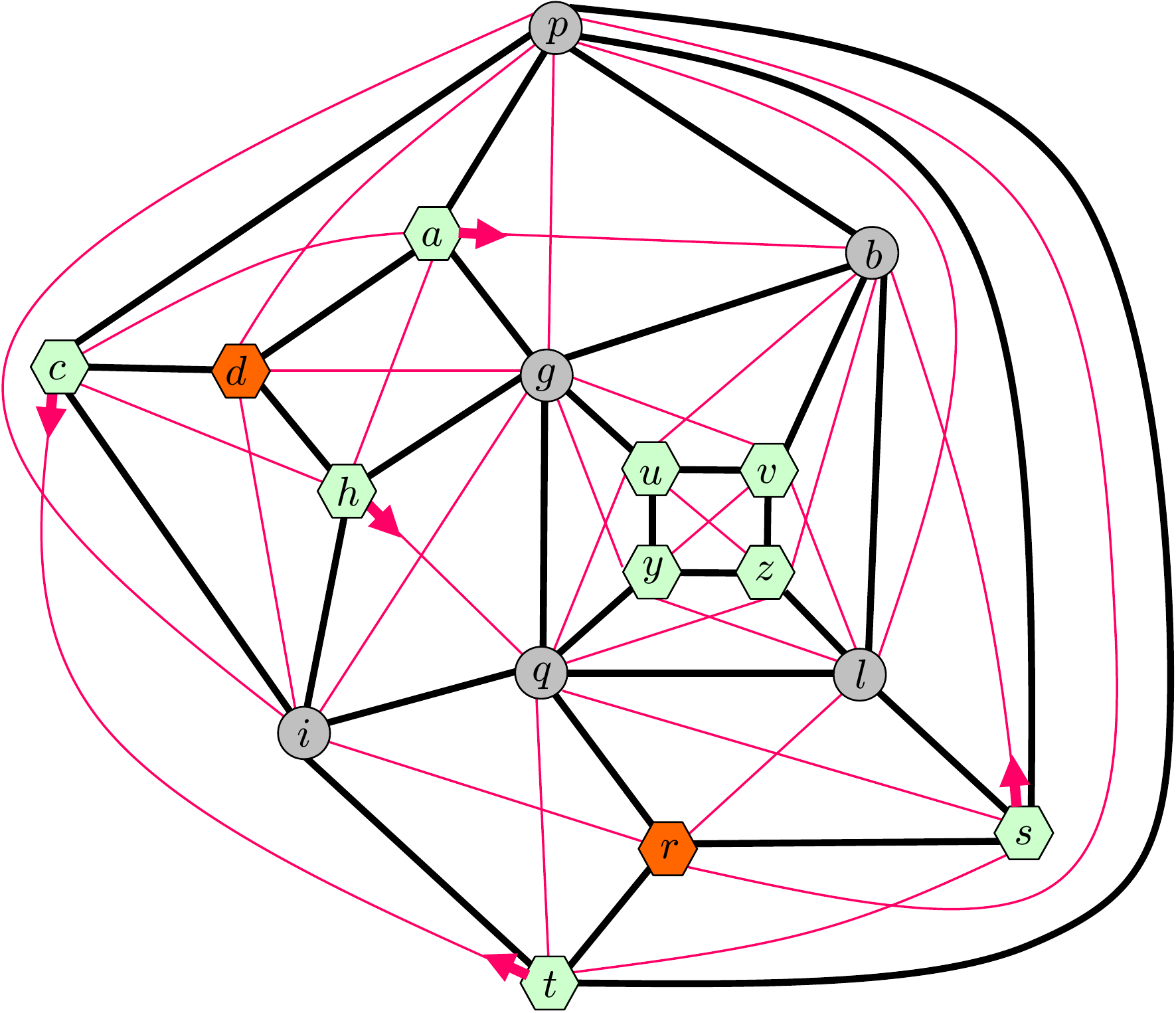}
    }
    \label{fig:base}
  }
  \subfigure[Graph $G_{17}$ after $CR(u,v,y,z)$. ]{
    \parbox[b]{4.5cm}{%
      \centering
      \includegraphics[scale=0.23]{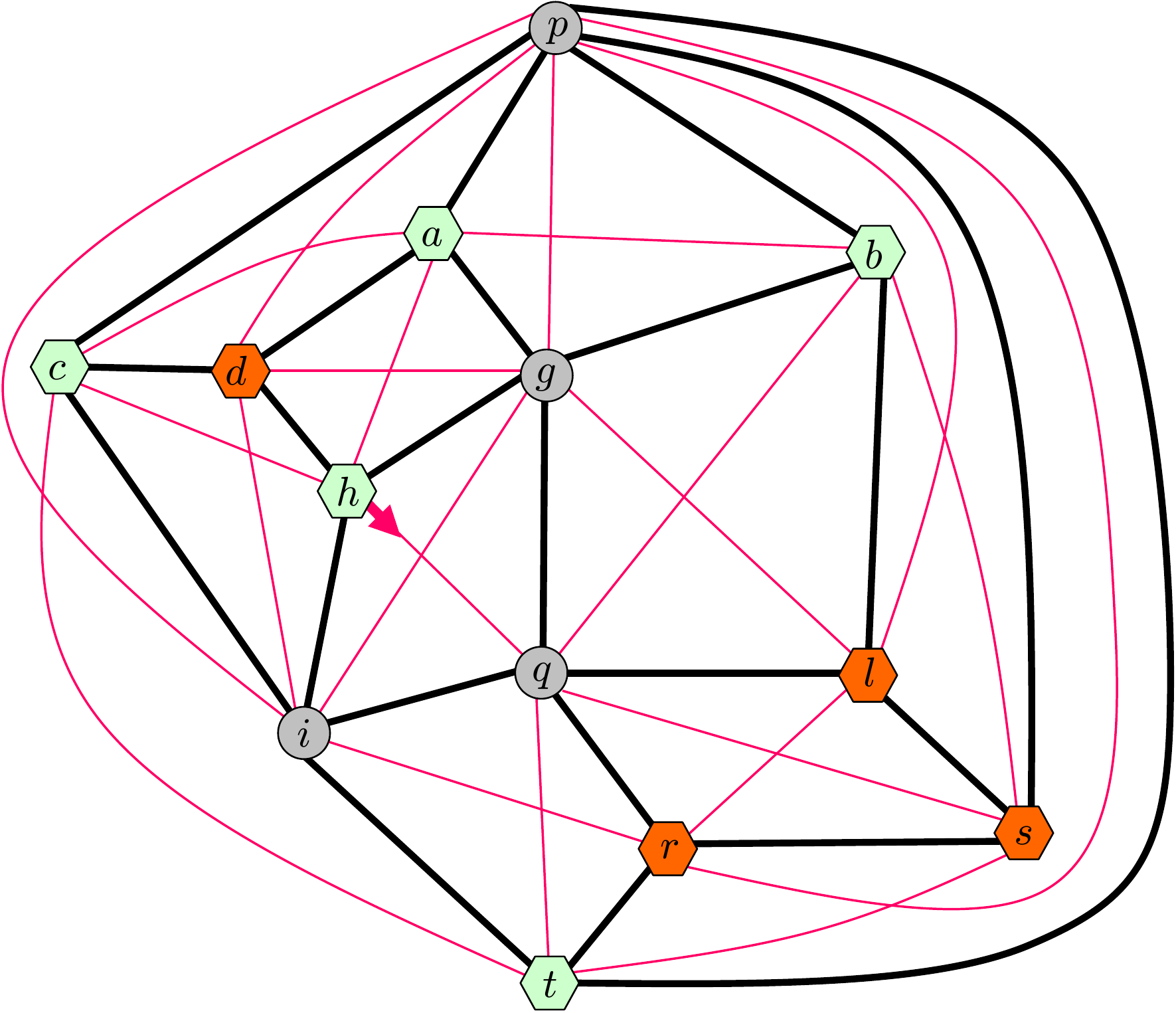}
    }
    \label{fig:G17CR}
  }
  \hfil
  \subfigure[and  $XW_{10}$ after $SR(h \mapsto q)$.]{
    \includegraphics[scale=0.23]{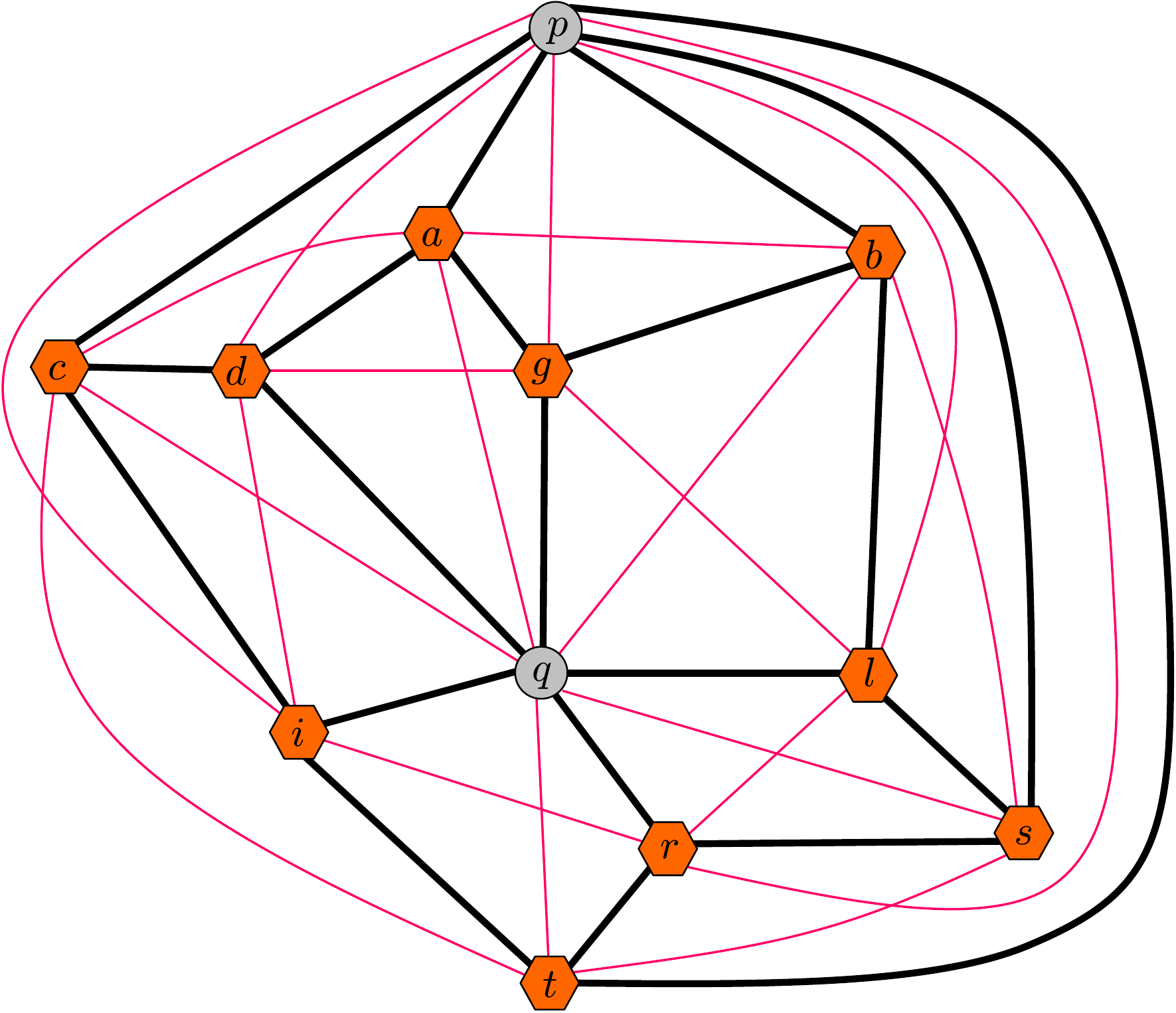}
    \label{fig:XW10}
  }
  \caption{A reduction of an input graph to an extended wheel graph.}
  \label{fig:reduce-G17}
\end{figure}

\begin{figure}
  \centering
  \subfigure[ after $SR(a \mapsto b)$]{
    \includegraphics[scale=0.23]{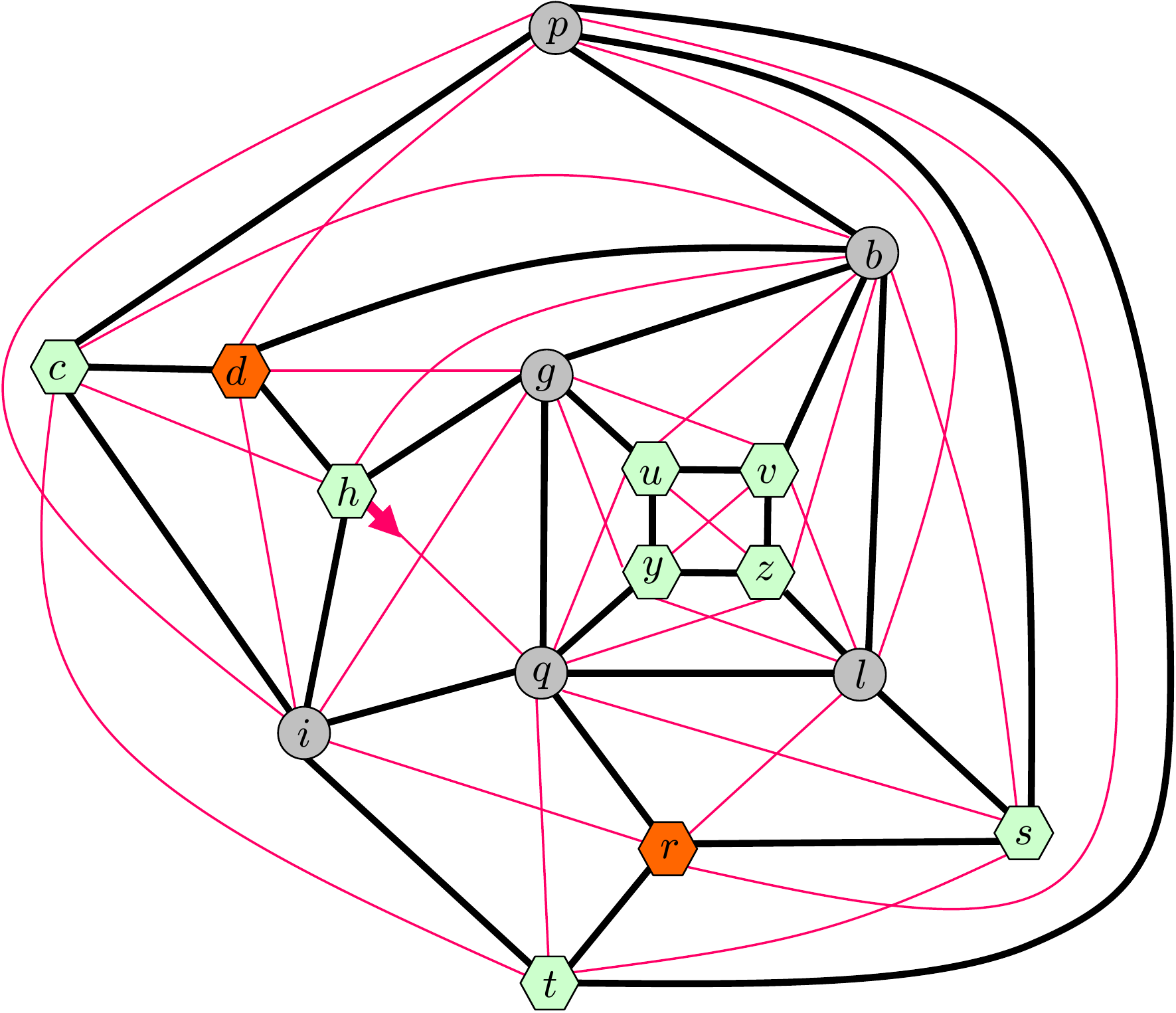}
    \label{fig:G17a}
  }
  \hfil
  \subfigure[and  after $SR(h \mapsto q)$ ]{
    \includegraphics[scale=0.23]{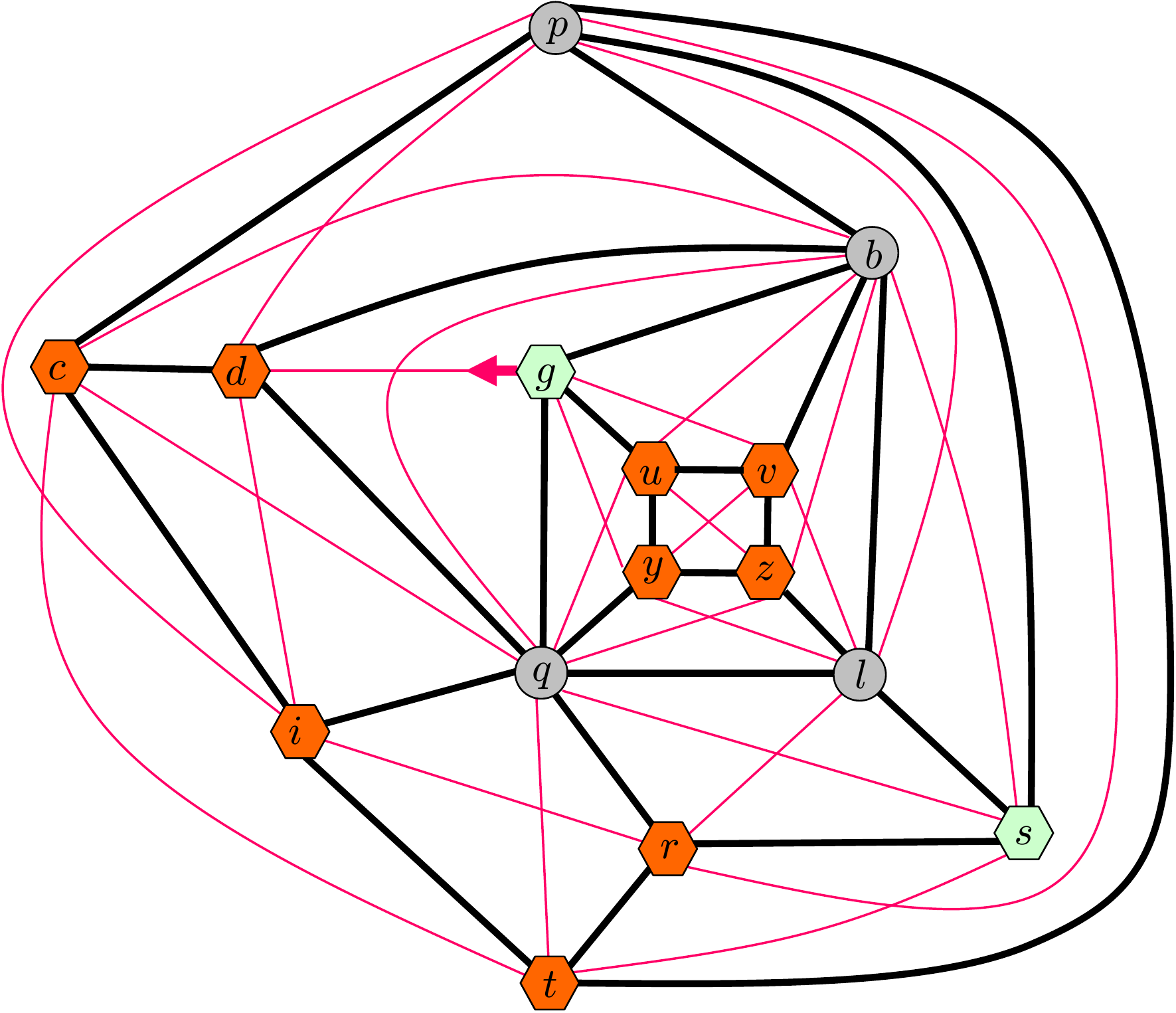}
    \label{fig:G17ah}
  } \\
  \subfigure[and  after $SR(g \mapsto d)$]{
    \includegraphics[scale=0.23]{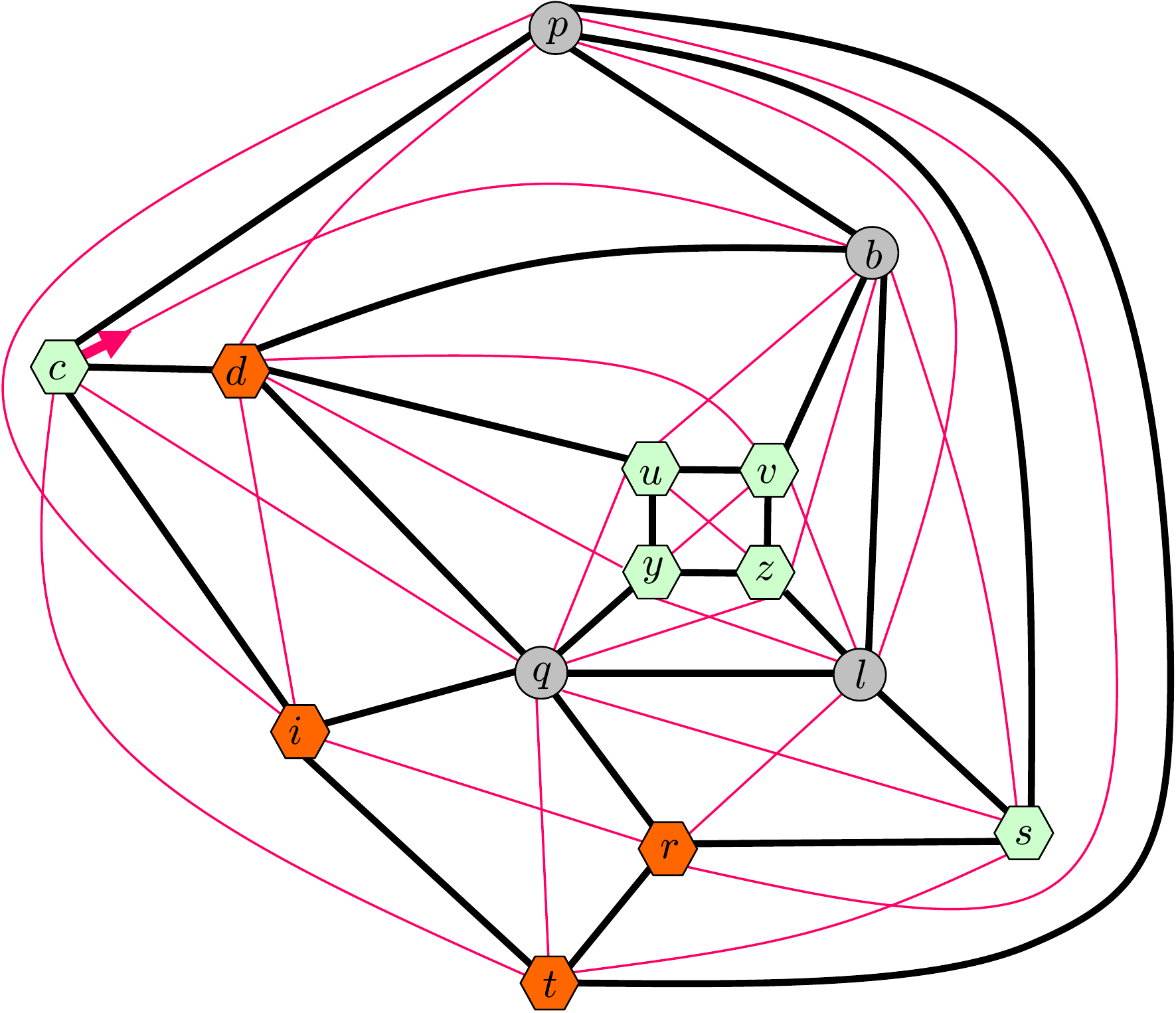}
    \label{fig:G17ahg}
  }
  \hfil
  \subfigure[and  after $SR(c \mapsto b)$]{
    \includegraphics[scale=0.23]{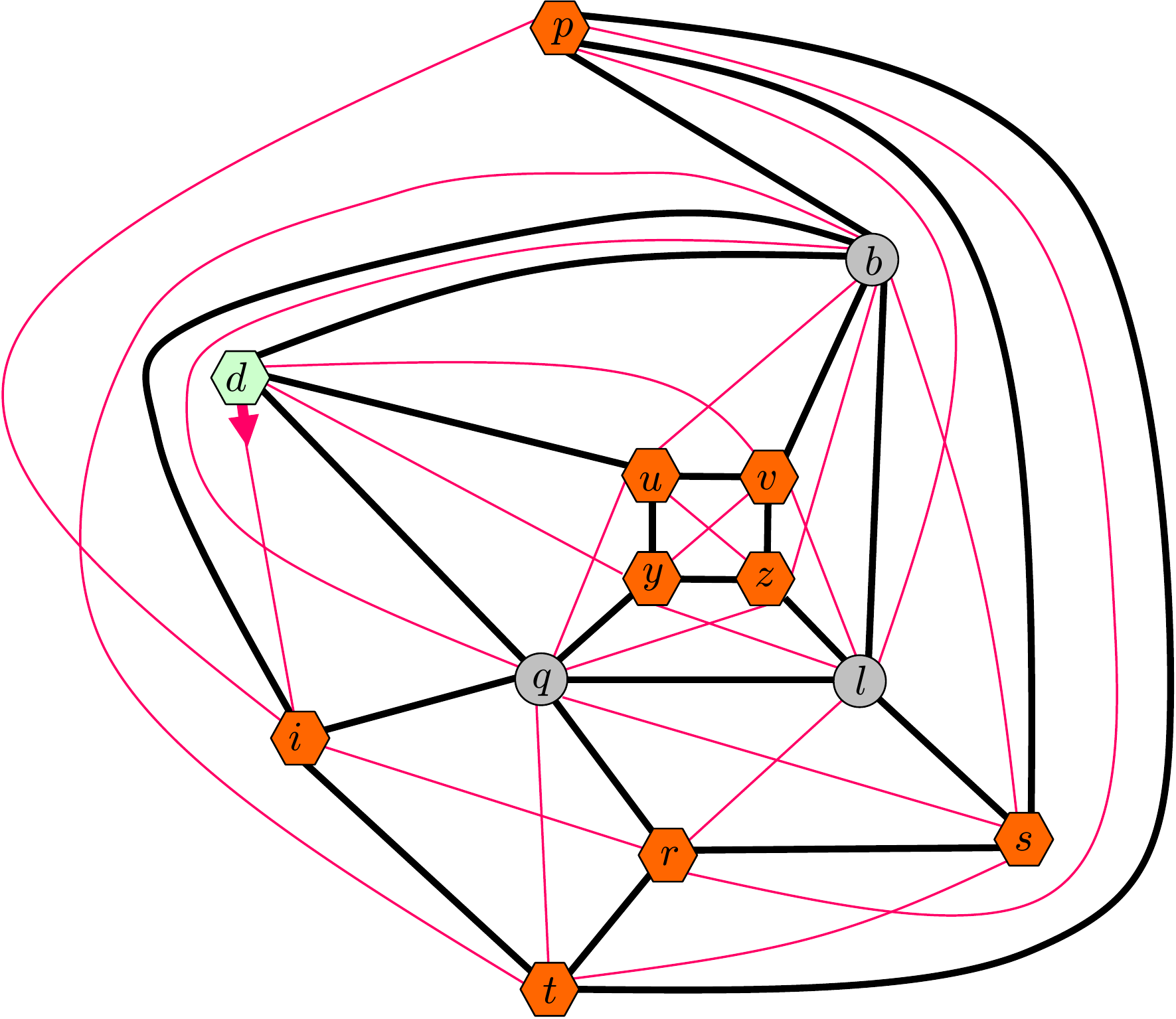}
    \label{fig:G17ahgc}
  } \\
  \subfigure[and after $SR(d \mapsto i)$]{
    \includegraphics[scale=0.23]{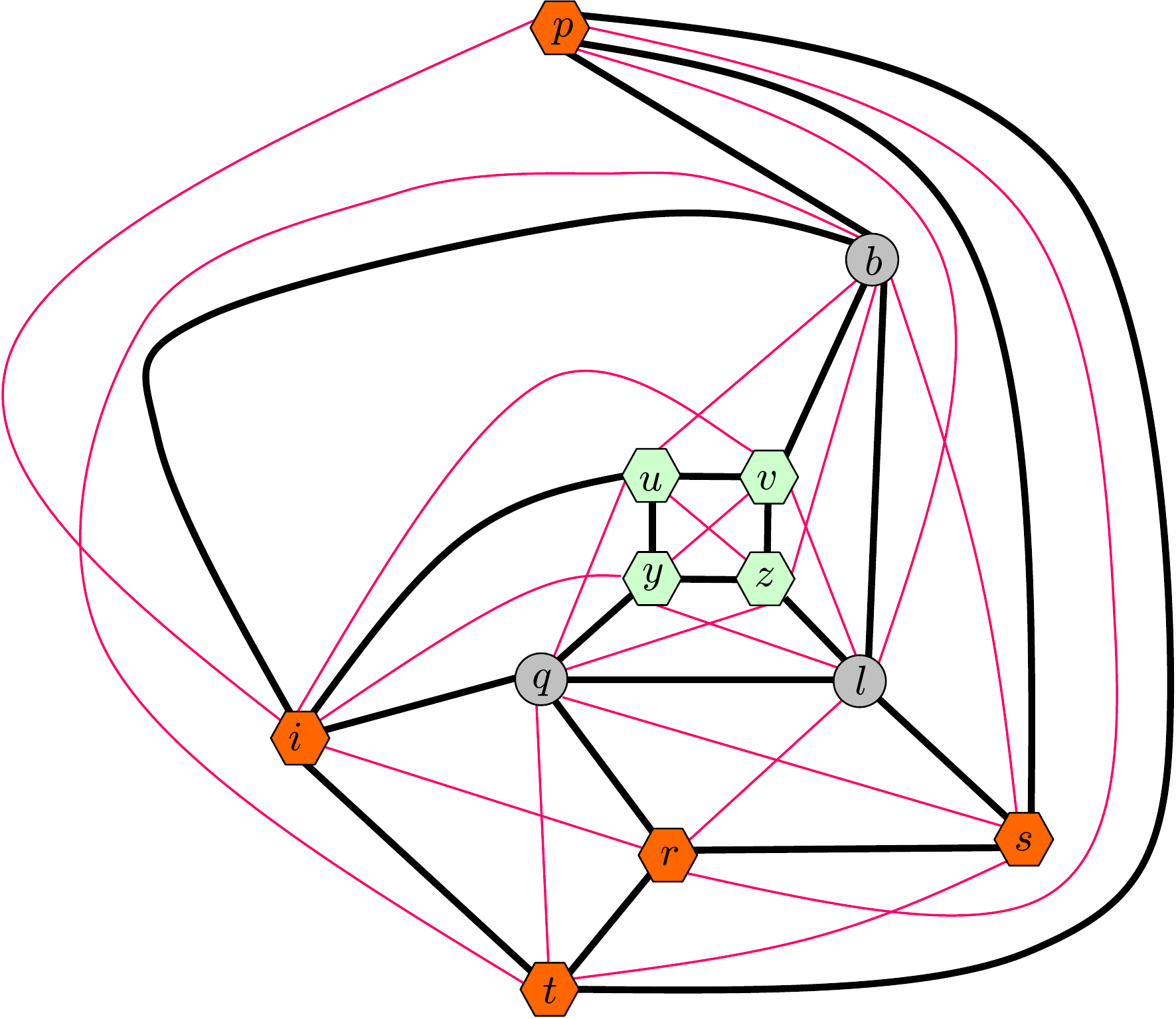}
    \label{fig:G17ahgcd}
  }
  \hfil
  \subfigure[and finally $CR(u,v,y,z)$  yields  $XW_6$]{
    \parbox[b]{5.0cm}{%
      \centering
      \includegraphics[scale=0.23]{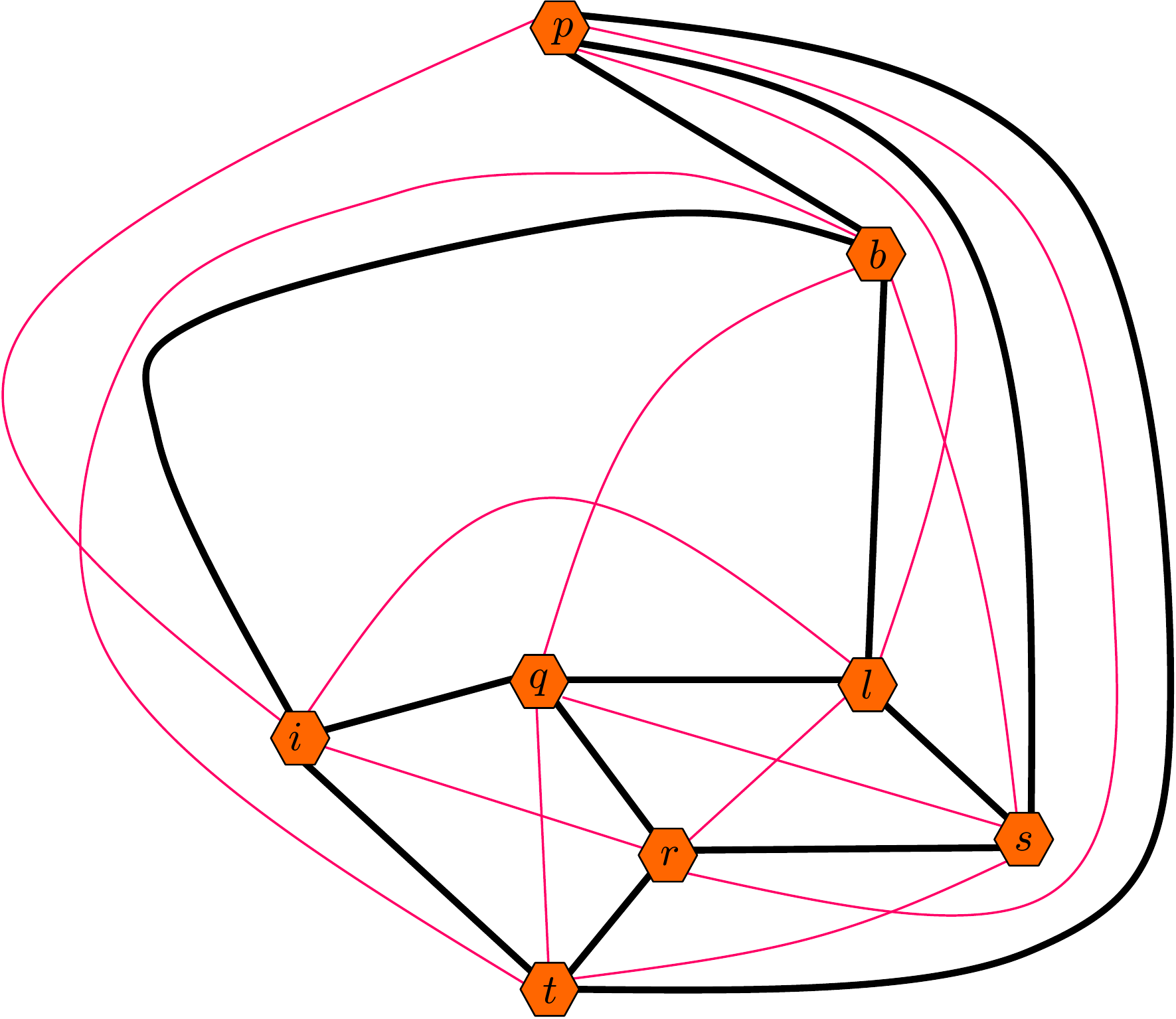}
    }
    \label{fig:G17ahgcdxy}
  }
  \caption{A alternative reduction of an input graph to $XW_6$.}
  \label{fig:reduction-XW8}
\end{figure}

Graph $G_{17}$ in Example \ref{ex:example1} can be reduced to
different extended wheel graphs  which are irreducible. In
consequence, the graph reduction system with the rules $SR$ and $CR$
cannot be confluent, since confluence implies a unique irreducible
representative.
 A rewriting
system is \emph{confluent} if $x \rightarrow^* u$ and  $x
\rightarrow^* v$ implies that there is a common descendant $z$ with
$u \rightarrow^* z$ and $y \rightarrow^* z$. In consequence, if two
rules  can be applied at different places of $x$  starting two
 reductions, then the reductions join at a common descendant.
 This is not the case for $SR$ and $CR$. More properties are
 elaborated in \cite{b-rso1p-16}.

\begin{corollary}
The reduction system with the rules $SR$ and $CR$ is non-confluent
on optimal 1-planar graphs.
\end{corollary}

\section{A Linear Time Algorithm} \label{sect:lineartime}

Example \ref{ex:example1} shows that a reduction  may change the
role of other candidates and reductions. In particular, $SR(x
\mapsto x_4)$ increases the degree of $x_4$ by two. If $x_4$ was a
candidate before, it is no longer. If $x_4$ blocked another
reduction, it does no longer. On the other hand, the degree of $x_3$
and $x_5$ decreases by two and they may become new candidates, which
may turn their neighbor candidates   from good to bad. Accordingly,
a $CR$-reduction decreases the degree of the vertices on the outer
cycle by two, which may introduce some of them as candidates with an
impact on candidates in their neighborhood. However, vertices at
distance at least three from the vertex of the application of a rule
are not affected.
  Hence, a reduction operates locally. However, a reduction may   have a global
effect and introduce or remove a blocking edge for many other
reductions and candidates. This is illustrated in Figs.\ \ref{multi}
and \ref{multired}. Thus it may be advantageous to maintain lists
with all   reductions that are or may be blocked by an edge. There
is a separating $4$-cycle if an edge blocks two or more reductions.
Clearly, there is a separating $4$-cycle at a $CR$-reduction. If two
$SR$-reductions are blocked by an edge $e=(u,v)$, then the centers
of the reductions have degree $6$ and have $u$ and $v$ as common
neighbors and there is a separating $4$-cycle through $u$ and $v$
which includes $e$ if the blocking edge is black.

We shall assume throughout that the given graph $G$ is reducible,
i.e., not an extended wheel graph, and that $x$ is a candidate of a
$SR$-reduction or $x$ is one of four candidates of a $CR$-reduction.

The degree vector of a candidate $x$ does not determine the coloring
of $H(x)$, but it tells which reduction is applicable, see Lemmas
\ref{localtest} - \ref{apply-CR}.  Infeasible applications can be
restricted even further.

 First, observe that $\tau(x) \leq 4$ if $x$ is a candidate of a reducible
 optimal 1-planar graph. Otherwise,
$\overrightarrow{H(x)} = (5,5,5,5,5,5,6)$  implies $H(x) = K_6$, but
$K_6$ is not a proper subgraph of a $5$-connected 1-planar graph
\cite{bsw-1og-84}.

Second, both blocking edges of a $CR$-reduction cannot occur
simultaneously, since the induced subgraph would have a separation
pair violating $4$-connectivity.

Finally, suppose there is a blocking black edge for an
$SR$-reduction at candidate $x$, say edge $(x_3, x_6)$. Then
$\overrightarrow{H(x)} = (3,3,4,5,5,6,6)$ or $\overrightarrow{H(x)}
= (3,4,5,5,5,6,6)$ by Lemma \ref{lem:H(x)}. If vertex $x_4$ has
local degree three, then $SR(x \mapsto x_4)$ is applicable, whereas
$SR(x \mapsto x_6)$ is not. Suppose the coloring is fixed as in
Fig.\ \ref{SR}, otherwise one must also consider cases with $x_1$
and $x_6$ exchanged.
Then $(x_1, x_2, x_3, x_6)$ and $(x_1, x_6, x_4, x_3)$ are
separating $4$-cycles which separates $x$ from further black
neighbors of $x_2$ and $x_4$, respectively. The blocking black edge
$(x_3, x_6)$ for $SR(x \mapsto x_6)$ cannot be removed if $x$
remains as a candidate. Then another reduction must remove the edge.
A black edge is removed by a  reduction if it is incident to a
candidate. However, if $x$ is a candidate and there is the edge
$(x_3, x_6)$, then $x_3$ has degree at least eight and the degree of
$x_3$ cannot be decreased to six since $x$ is a blocking neighbor.
Consider the $4$-cycle $C = (x_1, x_2, x_3, x_6)$ and suppose that
$x$ is in the outer face of $C$, the other case is similar. Then
$x_2$ has at least one more black neighbor $w$ besides $x_1$ and
$x_3$. If $x_6$ had local degree six, then the red edge $(x_6, w)$
were crossed by $(x_1, x_3)$, which were a multiple edge. Hence,
also $x_6$ has degree at least eight and is not a candidate. Hence,
the black prohibited edge $(x_3, x_6)$ remains if $x$ remains.
However, if $SR(x \mapsto x_2)$ or $SR(x \mapsto x_4)$ can be
applied, then $x$ is removed and also $SR(x \mapsto x_6)$. In
consequence,  a reduction $SR(x \mapsto v)$ can never be used if
there is a blocking black edge incident to $v$, and we add ``none''
in the last column of Table 1.

We summarize these facts:

\begin{lemma} \label{lem:fact}
For a reducible optimal 1-planar graph the following holds:

\begin{enumerate}
  \item If $x$ is a candidate, then $H(x)$ is fixed or
has a partial coloring for an $SR$-reduction.
  \item The subgraph matched by $CC$ has at most
  one   blocking red edge.
   \item A  reduction $SR(x \mapsto v)$ is infeasible if there is a
blocking black edge incident to $v$.
\end{enumerate}

\end{lemma}

Next, consider the interaction between $SR$- and $CR$-reductions.
Their usability is distinguished by the type of the candidates. The
vertices of the inner cycle of $CC$ mutually block each other for a
$SR$-reduction. These vertices are a ``black hole'' for
$SR$-reductions, since they can never take the role of the center of
a good $SR$-reduction. However,
 vertex $x$ of the inner cycle of $CC$ may be the target of a
$SR$-reduction $SR(w \mapsto x)$, whose use absorbs vertex $w$.
 In that case, the
$CR$-reduction is bad and is blocked by $w$. The vertices of the
inner cycle  can only be removed by a $CR$-reduction, or they remain
for the final extended wheel graph.

\begin{lemma} \label{lem:blackhole}
 A $SR$-reduction never applies to a candidate $x_i$ if
a $CR$-reduction applies to candidates $x_1, x_2, x_3, x_4$ for
$i=1,2,3,4$.
\end{lemma}

\begin{proof}
If $CR$ applies to $x_1, x_2, x_3, x_4$, then
$\overrightarrow{H(x_i)} = (4,4,5,5,5,5,6)$ for $i=1,2,3,4$ if the
reduction is good and   $\overrightarrow{H(x_i)} = (5,5,5,5,5,5,6)$
for two vertices if the reduction is bad  by Lemma \ref{apply-CR}.
The matching subgraph has a unique embedding. Vertices $x_i$ on the
inner cycle of $CC$ mutually block each other and $\tau(x_i) \geq 4$
excludes the use of a $SR$-reduction, which needs $\tau(x_i) = 3$ by
Lemma \ref{apply-SR}. \qed
\end{proof}

Finally, consider the relationship  between reductions and blocking
edges.
A reduction may introduce a blocking edge for many reductions, and
it may be blocked by several blocking edges. It is a many-to-many
relation, say $(j:k)$, where  $k$ may be linear in the size of the
graph. By Lemma \ref{lem:fact} it suffices to consider reductions
with a fixed or a partial coloring, and a reduction  with a blocking
black edge can be discarded. Hence, a candidate $x$ may allow for
three $SR$-reductions towards its red neighbors if $H(x)$ is fixed.
 Each $SR$-reduction has zero, one, or two
  blocking red edges, where zero means that the reduction is good.
Therefore, $j \leq 2$ suffices. There are two bad $SR$-reductions
for a candidate $x$ if there is a single blocking red edge, and $x$
is bad if and only if there are two blocking red edges or the graph
is $XW_6$.

A $SR$-reduction $SR(x \mapsto x_4)$ introduces the planar edge
$(x_1, x_4)$, which simultaneously may close many 4-cycles and then
may block many other candidates and their $SR$-reduction towards
$x_4$, see Fig. \ref{multi}. Similarly, edges $(x_2, x_4)$ and
$(x_6, x_4)$ or the diagonals in $CR$ may be blocking red edges for
many other reductions, as Fig. \ref{multired} illustrates. Such
edges may be removed by another reduction, and then they can
reappear after a further  reduction.

\begin{figure}
  \begin{center}
     \includegraphics[scale=0.3]{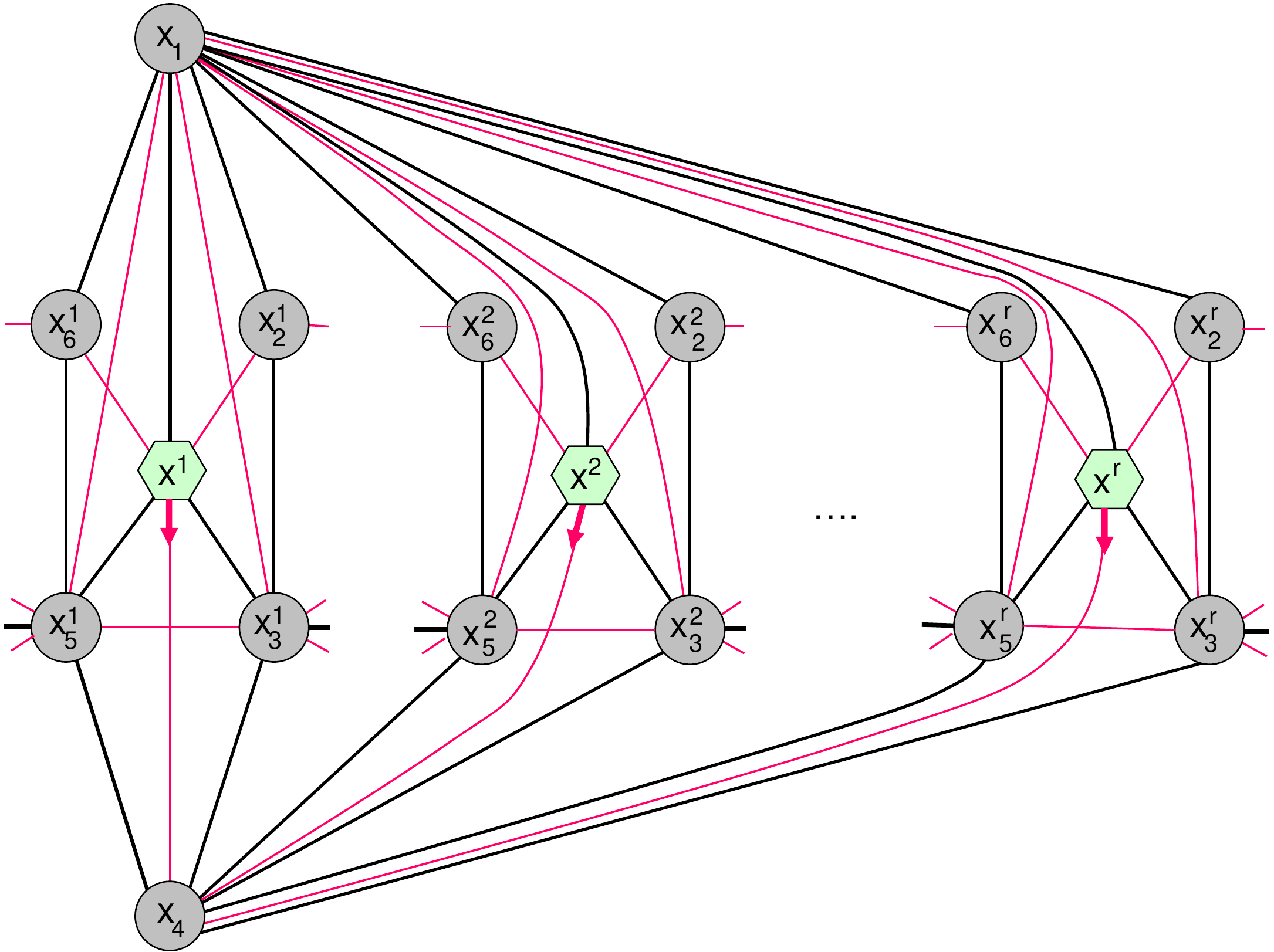}
     \caption{A conflict among   good candidates $x^1, \ldots, x^r$ and
       their $SR$-reduction $SR(x^i \mapsto x_4)$. The first such reduction
       generates a blocking black edge and blocks the other
       reductions. \label{multi}}
   \end{center}
\end{figure}

\begin{figure}
  \begin{center}
     \includegraphics[scale=0.30]{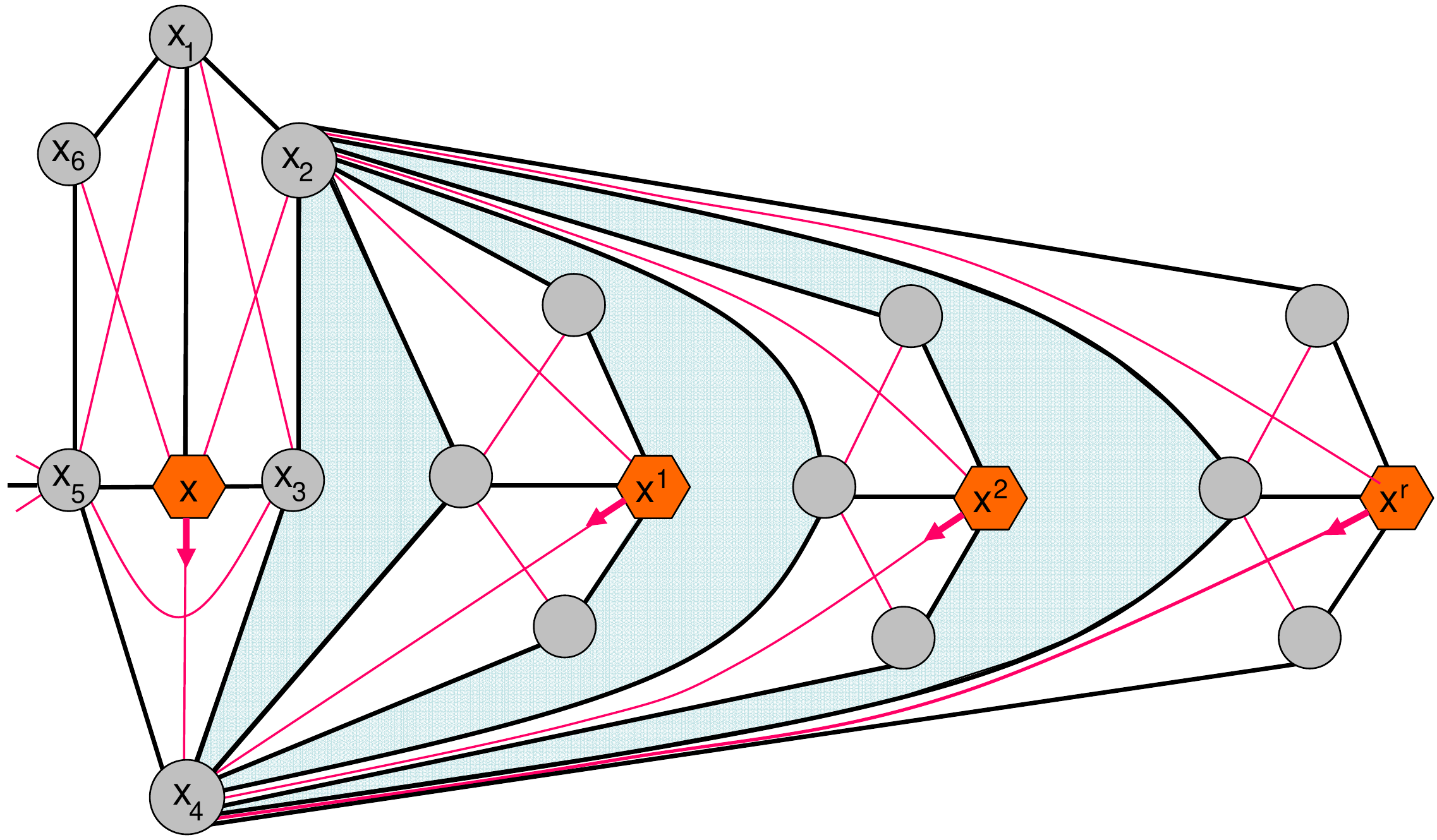}
     \caption{Illustration of a blocking red edge $(x_2, x_4)$ with many
       conflicts. There must be subgraphs in the shaded areas. If edge
       $e = (x_2, x_4)$ is introduced, e.g., by $SR(x \mapsto v_4)$, then $e$
       blocks the  reductions $SR(x^i \mapsto x_4)$ for $i = 1, \ldots, r$.
       If $e$ is removed thereafter, it can be reintroduced by
       $SR(x^j \mapsto x_4)$ for some $j$ and then it blocks the remaining
       $SR$-reductions again.\label{multired}}
   \end{center}
\end{figure}

A direct treatment of all candidates and their reductions may lead
to a quadratic running time. We use lists  for the reductions that
 a  red edge may block. For example, all reductions $SR(x^i \mapsto
 x_4)$ and $SR(x \mapsto x_4)$ in Fig.\ \ref{multired}
 are collected in a list $BAD_e$ if edge $e = (x_2, x_4)$ exists. The existence of a
 blocking red edge associated  with a $SR$-reduction is
determined by the degree vector and the local degree of the
vertices. The outcome is given in Table 1 and is a consequence of
Lemmas \ref{apply-SR}, \ref{apply-CR} and \ref{lem:fact}.

To manage the reductions efficiently, we split each pair of
associated blocking red edges of a reduction and treat each edge
separately. For each red edge that occurs in $H(x)$ for some
candidate $x$, there are three lists of reductions $GOOD_e, BAD_e$
and $WAIT_e$ and two entries of each reduction as given in Table 1.
Hence, there are up to six entries of $SR$-reductions at a
candidate. A reduction $\alpha$ is good if and only if $\alpha$ is
not blocked by an edge if and only if $\alpha$ is stored in $GOOD_e$
and  in $GOOD_f$ or $WAIT_f$ and $\alpha$ is not blocked by a black
edge. Here $e$ and $f$ are the blocking red edges associated with
$\alpha$. If $\alpha$ is blocked by $f$ and is not blocked by $e$,
then $\alpha$ is stored in $WAIT_e$ and in $BAD_f$ and, finally,
$\alpha$ is stored in $BAD_e$ and in $BAD_f$ if both associated
edges are blocking. In consequence, $BAD_e$ is empty if edge $e$
does not exist and, conversely, $GOOD_e$ and $WAIT_e$ are empty if
$e$ exists.

However, it may happen that   reduction $\alpha$ appears in $GOOD_e$
although $\alpha$ is blocked by the other blocking red edge $f$,
and, conversely, that $\alpha$ appears in $WAIT_e$ although $\alpha$
is good. This happens unnoticed to $e$ and $GOOD_e$ if edge $f$ is
(re)introduced or is removed, as indicated in Fig.\ \ref{multired}.
If $\alpha$ is accessed via $GOOD_e$ and $\alpha$ is bad, then there
is an \emph{unsuccessful access}, and $\alpha$ is moved from
$GOOD_e$ to $WAIT_e$.

$CR$-reductions have a higher priority than $SR$-reductions. If it
is encountered that a candidate $x$ has become a vertex of the inner
cycle of $CC$, then its $SR$-reductions are removed from the lists
 and are replaced by the $CR$-reduction. This
situation is detected as described in Lemma \ref{apply-CR} and is
justified by Lemma \ref{lem:blackhole}. In other words, $CR$
overrules $SR$.

In the next step of a computation a $\alpha$ is accessed via
$GOOD_e$ for some edge $e$. Then it is checked whether $\alpha$ is
good and if so,  $\alpha$ is applied  and some further actions are
taken.
Otherwise,  there is an \emph{unsuccessful access}. Then  $\alpha$
is  moved from $GOOD_e$   to $WAIT_e$ if there is the other blocking
red edge $f$, and $\alpha$ is removed from the lists if $\alpha =
SR(x \mapsto v)$ and there is a blocking black edge incident to $v$.

Suppose a reduction $SR(x \mapsto x_4)$ is good and is applied as
shown in Figs.\ \ref{SR} or \ref{SRupdate}.
 The case of a $CR$-reduction is similar, and even simpler. The actual graph is
modified as described by the $SR$-reduction. Vertex $x$ is removed
and so are all reductions at $x$ that are stored in the lists
$GOOD_e, BAD_e$,  and $WAIT_e$. Also all lists with a red edge
$e=(x,y)$ for some $y$ are removed. There are three vertices $y$,
since $x$ is a candidate, and these removals take  constant time. If
$x_4$ was a candidate before, all reductions at $x_4$ are removed,
since $x_4$ is no longer a candidate.

The $SR$-reduction removes edge $e=(x_3, x_5)$. Therefore, $BAD_e$
is renamed to $GOOD_e$. This makes the  stored reductions accessible
in the next step. Conversely,   $GOOD_e$ and $WAIT_e$ are renamed to
$BAD_e$ for $e = (x_2, x_4)$ and $e=(x_6, x_4)$, since these edges
are introduced and may be blocking red edges for other reductions.
Edge $h=(x_1, x_4)$ may become a blocking black edge, see Fig.\
\ref{multi}. Here, no action is taken and reductions blocked by $h$
are removed at an unsuccessful access or if one of $x_1$ or $x_4$ is
removed. Finally, vertices $x_3$ and $x_5$ may change their status
and become a candidate. We consider $x_3$; the case of $x_5$ is
similar. If vertex $x_3$ has become a candidate, then the possible
reductions on $x_3$ are computed and are added to the respective
lists $GOOD_e, BAD_e, WAIT_e$ and $GOOD_f, BAD_f, WAIT_f$ for the
pair of associated  red edges $e$ and $f$. Here $CR$ may overrule
$SR$.

A change of the status of $x_4$ to a non-candidate and of $x_3$ and
$x_5$ to a candidate has side effects on their neighbors if they
were candidates, too. This is illustrated by the color change of
candidates in Figs.\ \ref{fig:G17a} to \ref{fig:G17ahgcdxy} and in
Fig.\ \ref{SRupdate}. However, there is no need for a special
treatment, since everything is done by
 renaming  the lists.

\begin{figure}
  \begin{center}
     \includegraphics[scale=0.35]{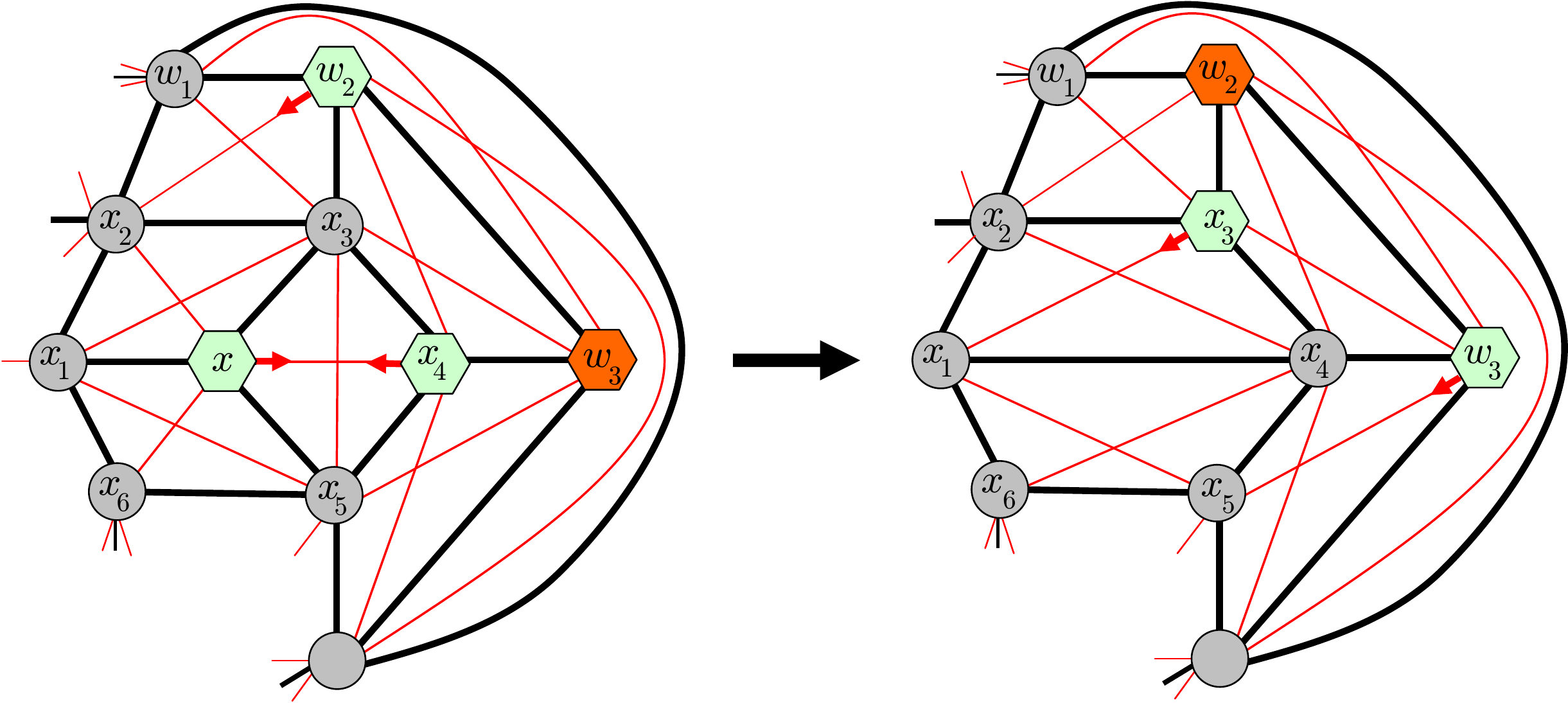}
     \caption{An update by $SR(x \mapsto x_4)$ with candidates in the neighborhood.
       \label{SRupdate}}
   \end{center}
\end{figure}

Our linear time algorithm  operates in three phases. First, it makes
a static check that all vertices of the input graph $G$ of size $n$
have   even degree at least six and that there are $4n-8$ edges.
Then it sweeps the given graph for candidates $x$, checks $H(x)$,
classifies and stores the reductions, and colors as many edges as
possible.  A second sweep may be helpful to clear some partial
colorings. In general, it creates six entries for the
$SR$-reductions at a candidate  $x$ and stores them in the lists
$GOOD_e, BAD_e$, and $WAIT_e$ for each associated blocking red edge
$e$. Two entries are discarded if there is a blocking black edge. If
there is a $CR$-reduction at $x$, then two entries are created and
  $SR$-reductions at $x$ are removed immediately.
  If, surprisingly, the coloring of $G$
is complete, we are done. The planar skeleton is 3-connected and has
a unique embedding and we test straightforwardly whether $G$ is
optima1 1-planar.  In general, there is a computation by a sequence
of steps and each step is a reduction on a presumably optimal
1-planar graph  $G=G_0 \rightarrow G_1 \rightarrow \ldots
\rightarrow G_t = XW_{2k}$ for some $k \geq 3$ and $t \geq 0$. The
algorithm immediately stops and reports a failure if the conditions
for the application of a reduction are not met or there is a
mismatch in the edge coloring between the graph and a reduction.

The algorithm has access to the lists $GOOD_e$, which, internally,
are combined to a superlist. The data structure resembles an
adjacency list for storing graphs. Empty sublists are removed. The
algorithm renames lists  which, internally, means removing and
inserting sublists  and takes $O(1)$ time.
There is no preference or restriction for the manipulation of the
superlist, which can be organized as a stack or as a queue or at
random. The next reduction is taken from the neighborhood of the
previous one if the superlist is organized as a stack, and all
candidates of a given graph are checked sequentially if there is a
queue. Moreover, one may use $CR$-reductions with higher priority
than $SR$-reductions, since they remove four vertices in a step and
have only two entries. Anyhow, there is a linear running time.

\begin{algorithm}\label{alg:a}
  \KwIn{A graph $G$}
  \KwOut{$G$ is optimal 1-planar or fail}
  \KwData{A collection of lists  $GOOD_e$, $BAD_e$ and $WAIT_e$ for some
    edges $e$. Each list contains reductions $\alpha = SR(x \mapsto x_4)$ or
    $\alpha = CR(u,v,y,z)$ with candidates $x, u, v, y, z$.}
  \BlankLine
  \algcomline{Preprocessing}\;
  %\tcc{Preprocessing}
  \lIf{$\neg$($G$ has $4n - 8$ edges and all vertices have even degree
      $\geq 6$)}{%
    \Return fail
  }
  \BlankLine
  \algcomline{Initialization}\;
  %\tcc{Initialization}
  \lForEach{candidate~$x$ of $G$}{
    call ADD\_REDUCTIONS
  }

  \BlankLine
  \algcomline{Processing}\;
  %\tcc{Processing}
  \While{there is a reduction $\alpha$ accessible via $GOOD_e$ for some edge
      $e$}{%
    \If{ $\alpha$  is good}{
      apply $\alpha$ and update $G$ and the coloring\;
      remove all reductions at  vertex $x$ from the lists $GOOD_e, BAD_e$
        and $WAIT_e$\;
      \If{$\alpha = SR(x \mapsto x_4)$}{%
        remove all reductions on $x_4$ from the lists $GOOD_e, BAD_e$ and
          $WAIT_e$\;
        remove all lists with $e = (x, y)$\;
        \lIf{ $x_3$ is a candidate after the $SR$-reduction}{%
          call ADD\_REDUCTIONS
        }
        \lIf{ $x_5$ is a candidate after the $SR$-reduction}{%
          call ADD\_REDUCTIONS
        }
        \lFor{$e = (x_3, x_5)$}{%
          rename $BAD_e$ to $GOOD_e$
        }
        \lFor {$e = (x_2, x_4)$ and $e=(x_6, x_4)$}{%
          rename $GOOD_e$ and $WAIT_e$ to $BAD_e$
        }
      }
      \Else(\algcom{$\alpha$ is a $CR$-reduction with outer cycle $(v_1, v_2, v_3, v_4)$}) {%
         apply $\alpha$ and update $G$ and the coloring\;
        \For{$v = v_1, v_2, v_3, v_4$} {%
          \lIf{$v$ is a candidate}{%
             call ADD\_REDUCTIONS
          }
        }
        \lFor {$e = (v_1, v_3)$  and $e= (v_2, v_4)$}{%
          rename $GOOD_e$ and $WAIT_e$  to $BAD_e$
        }
      }
    }
    \Else(\algcom{an unsuccessful access to a reduction}){%
      \lIf{$\alpha = SR(x \mapsto v)$ has a  blocking black edge incident to $v$}{
        remove $\alpha$ from the lists
      }
      \lElse{%
       move $\alpha$ from $GOOD_e$ to $WAIT_e$
      }
    }
}
  \lIf{$G$ is an extended wheel graph}{%
    \Return $G$ is optimal 1-planar
  }
  \lElse{%
    \Return fail
  }
  \caption{OPTIMAL 1-PLANARITY TESTING}
  \label{alg:optimal}
\end{algorithm}

\begin{algorithm}
  \KwIn{A candidate $x$ and lists $GOOD_e$, $BAD_e$ and $WAIT_e$ for some
    edges $e$}
  \KwOut{lists $GOOD_e$, $BAD_e$ and $WAIT_e$}
  \BlankLine
  \If{$x$ belongs to the inner cycle of $CC$ with outer cycle $(v_1, v_2,
      v_3, v_3)$}{
    \For{ $e = (v_1, v_3)$ and $f=(v_2,v_4)$}{%
      \lIf{$e$ exists}{%
        add $CR(x)$ to $BAD_e$ and to $WAIT_f$
      }
      \lElseIf{$f$ exists}{%
        add $CR(x)$ to $BAD_f$ and to $WAIT_e$
      }
      \lElse{%
        add $CR(x)$ to $GOOD_e$ and to $GOOD_f$
      }
    }
    remove all entries with $SR$-reductions of vertices from the inner
      cycle of $CC$ from the lists\;
  }
  \Else(\algcom{$x$ is a candidate for  $SR$-reductions}){
    \ForEach{$SR$-reduction $\alpha$ at $x$ with associated blocking
        red edges $e$ and $f$}{%
      \If{$\alpha = SR(x \mapsto v)$ is blocked by $e$}{%
        add $\alpha$ to $BAD_e$\;
        \lIf{$\alpha$ is blocked by $f$}{%
          add $\alpha$ to $BAD_f$
        }
        \lElse{%
          add $\alpha$ to $WAIT_f$
        }
      }
      \Else{%
        \lIf{$\alpha$ is blocked by $f$}{%
          add $\alpha$ to $WAIT_e$ and to $BAD_f$
        }
        \lElse{%
          add $\alpha$ to $GOOD_e$ and to $GOOD_f$
        }
      }
    }
  }
  \caption{ADD REDUCTIONS}\label{alg:optimal}
\end{algorithm}

Algorithm \ref{alg:a} preserves the following invariant:

\begin{lemma} \label{lem:properties}
Let  $G=G_0 \rightarrow G_1 \rightarrow \ldots \rightarrow G_t$ for
some  $t \geq 0$ be the sequence of graphs computed by the algorithm
on an optimal 1-planar graph $G$, i.e., a successful computation.
For every $i = 1, \ldots,t$ the following holds for $G_i$ and the
lists $GOOD_e, BAD_e$ and $WAIT_e$:

  \begin{enumerate}
  \item Each graph $G_i$ is optimal 1-planar and $G_t = XW_{2k}$
  for some $k \geq 3$.
  \item For each candidate $x$ of $G_i$ three  $SR$-reductions $\alpha_1, \alpha_2, \alpha_3$ at
    $x$ are each stored in the lists of  their associated
    blocking red edges if $x$ does not belong to an inner cycle
    of $CC$. If there is a blocking black edge, then only two
     $SR$-reductions may be stored; the one with an endvertex of
    the blocking black edge as a target may be missing.
  \item If $x$ belongs to an inner
    cycle of $CC$, then one entry of $CR$ is stored  in the
    lists of each associated blocking red edge.

  \item  If $\alpha$ is in $GOOD_e$ or in in $WAIT_e$, then $\alpha$ is not blocked by
  $e$.
  \item A reduction $\alpha$ is in $BAD_e$ if and only if $\alpha$ is blocked by   $e$.
\item A reduction $\alpha$ is good if and only if $\alpha$ is in $GOOD_e$ and in
$GOOD_f$ or $WAIT_f$ for the associated blocking red edges  $e$
    and $f$ and $\alpha$ is not blocked by a blocking black
    edge.

  \item  If there is an entry in the lists of edge $e$, then
  $BAD_e$ is nonempty if and only if $GOOD_e$ and $WAIT_e$ are empty.
  \end{enumerate}
\end{lemma}

\begin{proof}
The first property is due to Propositions \ref{prop3} and
\ref{prop4} since the algorithm either applies a reduction or does
not change the graph if there is an unsuccessful access. Properties
2 and 7 hold  for $G_1$ after the initialization, and they are
maintained by each successful reduction $G_i \rightarrow G_{i+1}$
for $1 \leq i < t$. If in the $i$-th step there is an unsuccessful
access to some reduction $\alpha$ in $GOOD_e$, then $\alpha$ is bad
and the red edge $e$ does not exist in $G_i$. Then $\alpha$ is
blocked  by a blocking black edge, in which case $\alpha$ is
removed, or by the other associated blocking red edge $f$, in which
case   $\alpha$ is moved from $GOOD_e$ to $WAIT_e$,  and the
invariant is preserved. \qed
\end{proof}

Concerning the running time, the critical part is the number of
unsuccessful accesses.

\begin{lemma}\label{lem:unsuccessful}
  If $G=G_0 \rightarrow G_1 \rightarrow \ldots \rightarrow G_t$ is a
  successful computation of Algorithm \ref{alg:a} on an optimal 1-planar
  graph $G$ of size $n$, then there are at most $O(n)$ unsuccessful
  accesses.
\end{lemma}

\begin{proof}
Clearly, there are at most $n$ successful reductions. First, there
are at most $3n$ unsuccessful accesses by blocking black edges,
since, in total, $G_1, \ldots, G_t$ have at most $3n$ black edges.
Graph $G$ has $2n-4$ black edges and each $SR$-reduction introduces
one black edge.

 Suppose, reduction $\alpha$ is accessed via
$GOOD_e$. Then $\alpha$ is not blocked by $e$ by Lemma
\ref{lem:properties}. If the access is unsuccessful by the other
associated blocking red edge, then $\alpha$ is moved to $WAIT_e$.
Suppose that $\alpha$ is accessed a second time via $GOOD_e$. Then
 $\alpha$ was moved from
$WAIT_e$ to $BAD_e$ when edge $e$ was inserted and from $BAD_e$ to
$GOOD_e$ when $e$ was removed by another reduction. Hence, there
were two successful reductions in between. As each edge may block
two $SR$-reductions, the number of unsuccessful reductions by
blocking red edges is bounded from above by the number of successful
reductions. In total, there are at most $4n$ unsuccessful
reductions.
 \qed
\end{proof}

In summary, we can state:

\begin{theorem}
  A graph $G$ is optimal 1-planar if and only if Algorithm \ref{alg:a}
  reduces $G$ to an extended wheel graph. If $G$ is optimal 1-planar, then
  a 1-planar embedding can be computed. The
  algorithm runs in linear time.
\end{theorem}

\begin{proof}
  The correctness follows  from Lemma \ref{lem:properties}. If $G$
  is reducible, then every reduction adds a partial embedding, which
  ultimately results in the unique embedding of $G$, otherwise,
  there is an embedding of an extended wheel graph.
  Clearly, the
  preprocessing and initialization phases take linear time, since each
  candidate and its reductions can be checked in constant time.
  Each successful
  reduction decreases the size at least by one and  takes $O(1)$ time, and there
  are  $O(n)$ unsuccessful accesses  by Lemma \ref{lem:unsuccessful}.
  Considering the maximum degree $d$ of a  vertex, it takes $O(1)$ time to test
  that a graph is not an extended
  wheel graph, since $d=n-2$ must hold for an optimal 1-planar graph of size $n$ \cite{s-s1pg-86}.
  Finally, the test for   $XW_{2k}$ takes $O(k$) time by Lemma \ref{XW-test}.
  Hence, each phase of the algorithm runs in linear time. \qed
\end{proof}

There is an immediate speed-up of the algorithm. If a reduction is
accessed, then it is checked whether the vertex of the reduction is
good. Thereby one considers three possible $SR$-reductions at a
time. Secondly, $CR$-reductions are preferred over $SR$-reductions,
since they remove four vertices at a time and lead to larger
extended wheel graphs and a faster termination of the algorithm.
Moreover, one can simplify the algorithm and avoid the bookkeeping
in lists if the graph is $5$-connected. Then the $SR$-reduction is
necessary and sufficient \cite{s-s1pg-86} and all updates are local.
The situations illustrated in Figs.\ \ref{multi} and \ref{multired}
cannot occur.

\begin{lemma}\label{4-cycle}
  There is a separating 4-cycle or a blocking vertex if a reduction is
  blocked by a (black or red) blocking edge.
\end{lemma}

\begin{proof}
  If $SR(x \mapsto x_4)$ is blocked by the black edge $(x_1, x_4)$, then
  it closes the 4-cycles $(x_1, x_2, x_3, x_4)$ and $(x_1, x_6, x_5, x_4)$,
  and these are separating, since they isolate $x$ from the further black
  neighbors of $x_3$ and $x_5$, respectively,  see Fig.\
  \ref{fig:blockingedges}.
  Accordingly, if there is a red edge $(x_2, x_4)$ and $x_3$ is not
  blocking, then there exist two vertices $u$ and $v$ such that the edge
  $(u,v)$ crosses $(x_2, x_4)$. Then $(x_2, x_3, x_4, u)$ and $(x_2, x_3,
  x_4, v)$ are separating 4-cycles isolating the further neighbors of $x_3$.

  If   edge $(v_1,v_3)$ is blocking for $CR$ with outer cycle
  $(v_1,v_2,v_3,v_4)$, then $(v_1,v_3)$ is red, since $(v_1,v_2)$ and
  $(v_2,v_3)$ are black by Lemma \ref{good-bad}. There is an edge $(u,v)$
  crossing $(v_1,v_3)$ if $v_2$ and $v_4 $ are not blocking. Then
  $(v_1,v,v_3,v_4)$ and $(v_1,u,v_3,v_4)$ are separating 4-cycles. \qed
\end{proof}

Schumacher \cite{s-s1pg-86} has shown that every $5$-connected
optimal 1-planar graph can be reduced to an extended wheel  graph
using only $SR$-reductions. Conversely, a $CR$-reduction must be
used if there is a separating $4$-cycle.

\begin{corollary}
There is a linear-time algorithm to test whether a graph is a
$5$-connected optimal 1-planar graph.
\end{corollary}
\begin{proof}
We restrict Algorithm 1 to use only $SR$-reductions and it succeeds
if and only if the given graph is a $5$-connected optimal 1-planar
graph. \qed
\end{proof}

\section{Conclusion and Perspectives}

We have added optimal 1-planar graphs to a list of graphs that can
be recognized in linear time. The restriction to optimal graphs is
important, since 1-planarity is \NP-hard, in general.

The algorithm shows that graph $B_{17}$ in Fig.\ \ref{fig:B17}  is
not optimal 1-planar. The graph is obtained from graph $G_{17}$ in
Fig.\ \ref{fig:base} by exchanging edges $(p,s), (c,h)$ and
$(p,h),(c,s)$. Consider candidate $c$ in $B_{17}$. Then $H(c) =
(2,3,4,4,4,5,6)$ violates optimal 1-planarity.

 Combinatorial properties of the $SR$- and $CR$-reductions have
 been studied in
\cite{b-rso1p-16}, where we have shown that every reducible optimal
$1$-planar graph $G$ can be reduced to every extended wheel graph
$XW_{2k}$ for $s \leq k \leq t$, where $s=3$ if and only if $G$ has
a separating $4$-cycle and $s=4$ if and only if $G$ is $5$-connected
and some $t < n$ for graphs of size $n$. The reductions to the small
extended wheel graphs can  also be computed in linear time.

The recognition problem of beyond planar  graphs is \NP-hard, in
general. It is open, whether there are other classes of optimal
graphs with a linear time
  recognition, e.g., optimal IC planar graphs with $\frac{13}{4}n - 6$ edges
  where each vertex is incident to at most one crossing edge \cite{bdeklm-IC-15} or optimal
  2-planar graphs with $5n-10$ edges, where kites from optimal 1-planar
  graphs are replaced by pentagons of $K_5$'s \cite{pt-gdfce-97}.

\subsubsection{Acknowledgements}

I would like to thank Christian Bachmaier and Josef Reislhuber for many
inspiring discussions and their support and the reviewers for their valuable
suggestions.

%% --------------------------------------------------------------------
%       Bibliography
%% --------------------------------------------------------------------

%\bibliographystyle{splncs03}
%\bibliographystyle{spmpsci}
%\bibliographystyle{abbrvurl}
\bibliographystyle{abbrv}
\bibliography{brandybib}

\end{document}